\documentclass[journal]{IEEEtran}
\usepackage{graphicx}
\usepackage{amsmath}
\usepackage[noend]{algpseudocode}
\usepackage{algorithmicx,algorithm}
\usepackage{amsthm}
\usepackage{amsfonts}
\usepackage{subfigure} 
\usepackage{color}
\usepackage{cite}
\usepackage{setspace}
\newtheorem{Lemma}{Lemma}
\newtheorem{Theorem}{Theorem}
\usepackage{amssymb}
\usepackage{color}
\graphicspath{{images/}}
\usepackage{booktabs}
\usepackage{multirow} 
\usepackage{makecell}
 \usepackage[numbers,sort&compress]{natbib}
 \newtheorem{remark}{Remark}
\usepackage{stfloats}
\usepackage{amsmath}
\usepackage{amssymb}
\usepackage{enumitem}
\usepackage{threeparttable}
\usepackage{comment}
\usepackage{bm}
\begin{document}	
\title{Proper Sea Surface Roughness Enhances the Performance of Near-Shore Maritime Networks}
	
\author{Wen-Yu Dong, Shaoshi Yang,~\IEEEmembership{Senior Member,~IEEE}, Song Zhao, Jinyang Yu, \\ Weiliang Xie, Rui-Si Han, Qi Bi,~\IEEEmembership{Fellow,~IEEE}, and Sheng Chen,~\IEEEmembership{Life Fellow,~IEEE}	%
\thanks{W.-Y. Dong, S. Zhao, J. Yu, and Q. Bi are with the Future Technology Research Center, China Telecom Research Institute, Beijing 102209, China (E-mails: dongwy@chinatelecom.cn; zhaosong1@chinatelecom.cn; yujy10@chinatelecom.cn;  qibi@chinatelecom.cn)}
\thanks{S. Yang is with the School of Information and Communication Engineering, Beijing University of Posts and Telecommunications, and also with the Key Laboratory of Mathematics and Information Networks, Ministry of Education, Beijing 100876, China (E-mail:  shaoshi.yang@bupt.edu.cn).}
\thanks{W. Xie is with the Mobile and Terminal Technology Research Department, China Telecom Research Institute, Beijing 102209, China (E-mail: xiewl@chinatelecom.cn)}
\thanks{R.-S. Han is with the Cloud Network Operating System R\&D Center, China Telecom, Beijing 102209, China (E-mail: hanruisi@chinatelecom.cn)}
\thanks{S. Chen is with the School of Electronics and Computer Science, University of Southampton, Southampton SO17 1BJ, U.K. (E-mail: sqc@ecs.soton.ac.uk).} %
\vspace*{-5mm}
}

%	\markboth{Journal of \LaTeX\ Class Files,~Vol.~14, No.~8, August~2021}%
%	{Shell \MakeLowercase{\textit{et al.}}: Stochastic Geometry Based Modeling and Analysis of Cooperative Satellite-Aerial-Terrestrial Systems for Nomadic Communications with Weak Signal Coverage}
	
	%	\IEEEpubid{0000--0000/00\$00.00~\copyright~2021 IEEE}
	% column for its text to clear the IEEEp  ubid mark.
	
\maketitle 

\begin{abstract}
	Accurate performance analysis for near-shore maritime wireless communication is essential for ensuring robust and reliable operations. However, existing analytical models often rely on oversimplified propagation assumptions, such as a perfectly smooth sea surface, which fail to capture the full dynamics of the maritime channel. In this paper, we develop a physically grounded analytical framework using stochastic geometry that bridges this gap. The spatial distribution of vessels is modeled as a non-homogeneous Poisson point process to reflect realistic near-port densities. We replace the idealized smooth-sea assumption by deriving a novel reflection coefficient from the classical Rayleigh criterion, which explicitly links the path loss to the significant wave height. Integrating this roughness-aware channel model into the stochastic geometry framework, we derive new analytical expressions for the uplink coverage probability and average ergodic rate, providing the first tractable characterization of aggregate interference under such dynamic conditions.
The analysis reveals a sea-state-dependent reliability--capacity trade-off: roughness-induced attenuation of the coherent specular reflection can suppress destructive-interference nulls and improve reliability-oriented coverage, while reducing high-SINR and average-rate performance. Available measurements support the underlying roughness-sensitive reflection mechanism, but direct VHF validation under rough sea conditions remains unavailable; the corresponding rough-sea results are therefore interpreted as model-based predictions. A cross-frequency ablation further confirms the wavelength dependence of the roughness effect and shows that the reflection coefficient must be evaluated for the operating frequency.
\end{abstract}
	

\section{Introduction}\label{S1}

\IEEEPARstart{T}{he} escalating demand for maritime autonomy, intelligent transportation, and robust safety services has underscored the critical importance of high-performance wireless communication in near-shore environments \cite{Alqurashi2023, Lyu2021, Ding2026ISAC}. Since the reliability of these data-centric applications hinges on a predictable and resilient physical layer, a thorough understanding and accurate modeling of the maritime propagation channel is an indispensable foundation for the effective analysis, design, and optimization of these vital communication systems \cite{Lin2020}.

Despite its importance, the near-shore maritime channel presents a formidable modeling challenge. Its character is dominated by complex multipath propagation over the sea surface \cite{WangAceess}, yet a pivotal deficiency in existing analytical frameworks is the oversimplified treatment of this surface. The prevalent assumption of a perfectly smooth, mirror-like (i.e., specular) surface is physically unrealistic. This idealization fails to capture the profound impact of sea state on reflection characteristics, leading to significant discrepancies between theoretical predictions and real-world channel behavior, particularly in the characterization of multipath fading depths. Addressing this physical-layer inaccuracy is the first critical step, but it is insufficient on its own. Modern maritime systems are not single links but complex networks of randomly located transceivers. Therefore, a comprehensive analysis demands a framework that can handle this network-scale randomness.

To meet this dual challenge, our work proposes a synergistic approach that integrates two powerful technical components. At the physical link level, we move beyond simplistic assumptions by incorporating a more realistic, physics-based propagation model. We build upon foundational work like the piecewise path loss model \cite{Lee20214} but enhance it to account for sea surface roughness. At the network level, we employ stochastic geometry \cite{1995Stochastic}, a potent and tractable methodology for analyzing large-scale wireless networks with randomly located nodes. The technical focus of this work is the integration of wavelength-dependent coherent-reflection attenuation with a non-homogeneous maritime network model. This integration makes it possible to trace how link-level multipath changes affect aggregate interference, coverage probability, and average rate.
The framework is intended for analytical evaluation under its stated propagation and access assumptions rather than as a complete protocol-level or measurement-calibrated representation of every operational maritime scenario.

\subsection{Related Works}\label{S1.1}

The maritime propagation environment presents a formidable modeling challenge, governed by the complex interplay between multipath reflections from the time-varying sea surface and anomalous atmospheric phenomena like evaporation ducts. To address this, the literature has evolved along two primary trajectories: high-fidelity physical simulations and tractable analytical models. The first trajectory prioritizes physical fidelity through computationally intensive methods. Notable examples include Ding et al. \cite{Ding2019}, who employed stochastic ray tracing to model diffuse sea surface reflections; Raulefs et al. \cite{Raulefs2023}, who introduced a geometric stochastic channel model for specific links incorporating terrestrial scatterers; and He et al. \cite{He2022}, who developed a non-stationary GSCM to jointly capture sea scattering and atmospheric ducting. While these simulation-centric frameworks provide exceptional link-level detail, they share a fundamental limitation: their inherent complexity precludes the derivation of a closed-form path loss expression, rendering them unsuitable for the tractable performance analysis of large-scale networks.

The second trajectory of research, which seeks analytical tractability, has produced several environment-aware propagation and coverage models \cite{Li2025LowAltitude}, though often with significant limitations in their physical grounding and applicability. For instance, Mehrnia and Ozdemir \cite{Mehrnia2016} proposed an empirically corrected two-ray model for millimeter-wave frequencies to fit ray-tracing results. However, this model is fundamentally limited as it neglects key physical phenomena like atmospheric ducting and sea surface roughness, and its correction factor remains purely empirical and lacks validation against measurement data. Similarly, Matolak and Sun \cite{Matolak2017} developed a measurement-based random model for air-to-ground channels that incorporates an intermittent third path. But this model is not directly applicable to our near-shore scenario and it also neglects the crucial effect of atmospheric ducting. The work by Lee et al. \cite{Lee20214} provided the seminal, measurement-validated piecewise model that transitions from a two-ray to a three-ray model to account for atmospheric ducting at longer ranges. Despite its significance, this model was validated under calm sea conditions and is based on ideal specular reflection, meaning it does not offer a physical parameterization of the sea state's impact on channel characteristics.

To leverage such path loss models for the system-level analysis of large-scale networks, where performance is fundamentally limited by interference arising from the random spatial locations of nodes \cite{Haenggi2009}, stochastic geometry has emerged as an indispensable tool. Seminal works \cite{Andrews2011, Dhillon2012} demonstrated that by modeling base station locations as a Poisson point process (PPP), it is possible to derive tractable and accurate expressions for key performance metrics such as coverage probability and average rate, even in complex multi-tier heterogeneous cellular networks. This methodology has since become a mainstream pillar of wireless research, proving its versatility in analyzing increasingly complex scenarios \cite{Zhao2019}, including cooperative non-terrestrial networks composed of satellites and unmanned aerial vehicles \cite{DongTCOM, DongJSAC}. More recently, fluid-spatiotemporal stochastic geometry has been introduced to characterize information flow in non-stationary spatial fields \cite{Dong2026FSTSG}. Building on this continuum perspective, related fluid-based formulations have been developed for flux-aware infrastructure provisioning in UAV logistics networks \cite{Dong2026TMC} and for vorticity-dissipation-based loop-free routing in ultra-dense networks \cite{Dong2026VDR}. These developments further extend spatial network modeling from static node distributions toward the dynamics and transport of network flows. Nevertheless, the unique combination of propagation conditions and non-uniform spatial distributions encountered in near-shore maritime networks remains insufficiently explored.

The choice of a large-scale path loss model is critical for accurately analyzing near-shore communication networks, and yet many system-level analyses adopt simplified models that lack physical fidelity.  A principal limitation of these works is their focus on single-link,  noise-limited performance, which overlooks the dominant role of co-channel interference in a shared-spectrum system. For instance, the seminal work on space-air-ground-sea integrated networks by Xu et al. \cite{Xu2023} employs the empirical model recommended by the International Telecommunication Union. While useful for high-level prediction, this black-box model is derived from curve-fitting and offers little insight into the underlying propagation mechanisms. Similarly, recent studies leveraging stochastic geometry for shore-to-ship analysis, such as the LEO satellite-aided communication framework \cite{Hu2024}, model the direct marine link with a standard power-law path loss. This simplification, while mathematically convenient, fails to capture the unique propagation characteristics of the near-shore maritime environment, which are dominated by sea-surface reflections and atmospheric ducting, particularly at low antenna heights.

Beyond the fidelity of the path loss model, a second critical disconnect between theory and practice lies in the spatial modeling of maritime traffic.  The few shore-to-ship analyses leveraging stochastic geometry often adopt overly simplistic spatial models, such as analyzing a single representative link or assuming a uniform vessel distribution \cite{Xu2023, Hu2024}. These assumptions, while tractable, fail to capture the large-scale, non-uniform nature of maritime traffic. In parallel, research in other wireless networking domains has addressed non-homogeneity using structured point processes, such as Poisson cluster processes \cite{Chun2015, DongGC, Cho2013}. However, these models are ill-suited for the arbitrary and unstructured patterns characteristic of near-shore vessel distributions. This analytical gap, compounded by the prevalent use of simplified path loss models that lack physical fidelity, underscores the absence of a comprehensive framework for accurately analyzing large-scale maritime communication systems.

\subsection{Our Contributions}\label{S1.2}
This paper develops a roughness-aware stochastic-geometry framework for large-scale near-shore maritime networks. Its focus is the cross-layer integration of wavelength-dependent coherent-reflection attenuation, a non-homogeneous vessel distribution, and aggregate uplink interference. The resulting framework connects maritime multipath propagation and spatial traffic non-uniformity to network-level coverage and rate. The measurement comparisons are interpreted within their respective frequency and sea-state scopes and are not used to claim complete empirical validation across all operating conditions. The main contributions are summarized as follows.
\begin{itemize}
	\item We establish a cross-layer analytical framework that combines a non-homogeneous Poisson point process for spatially varying vessel traffic with a roughness-aware piecewise two-ray/three-ray propagation model. Unlike analyses based on a representative link, a uniform vessel field, or a fixed power-law path-loss model, the proposed framework carries both spatial non-homogeneity and maritime multipath characteristics into system-level performance evaluation.
	\item Drawing on the classical Rayleigh roughness criterion, we formulate a wavelength-dependent effective reflection coefficient for the piecewise maritime path-loss model. The coefficient depends explicitly on the sea-state parameter, grazing geometry, and carrier wavelength, while recovering the classical smooth-sea model as the electrically smooth limiting case. This formulation enables the coherent specular reflection component to be evaluated independently at each operating frequency rather than transferred between C-band and VHF.
	\item We derive tractable expressions for the aggregate-interference Laplace transform, uplink coverage probability, and average ergodic rate. The analysis jointly incorporates the non-homogeneous vessel distribution, piecewise maritime path loss, Nakagami fading, and co-channel interference, thereby propagating link-level changes in coherent reflection to network-level reliability and rate.
	\item We evaluate the framework using complementary propagation and network-level evidence. The 5.15~GHz measurements assess roughness-sensitive attenuation of the coherent specular reflection component, whereas the 160~MHz measurements examine the electrically smooth VHF limit under calm-sea conditions. The 162~MHz rough-sea curves are presented as wavelength-specific model predictions. A cross-frequency coefficient-transfer ablation further examines the impact of substituting the C-band coherent-reflection coefficient for the wavelength-specific VHF coefficient while retaining all other VHF propagation parameters. The network-level analysis characterizes the conditional reliability--capacity trade-off associated with coherent-reflection attenuation under the considered vessel distribution and propagation geometry.
\end{itemize}

The remainder of this paper is organized as follows. Section~\ref{S2} details the system model, including the vessel distribution and channel characteristics. Section~\ref{S3} presents the refined roughness-aware path-loss model. Sections~\ref{S4} and~\ref{S5} derive the uplink coverage probability and average ergodic rate. Section~\ref{S6} presents numerical results, measurement comparisons within their stated frequency and sea-state scopes, and the cross-frequency transfer ablation. Section~\ref{S7} concludes the paper.

\textbf{Notation:} $\mathbb{P}(\cdot)$ denotes the probability measure and $\mathbb{E}[\cdot]$ denotes the expectation operator. The Laplace transform of random variable $X$ is defined as $\mathcal{L}_X\left(s\right)
=\mathbb{E}\,\left[\exp (-s X)\right]$. The cumulative distribution function (CDF) and probability density function (PDF) of $X$ are denoted by $F_X(x)$ and $f_X(x)$, respectively, while the conditional PDF of $X$ conditioned on $Y$ is denoted as $f_{X|Y}(x|y)$. $\Gamma(\cdot)$ is the Gamma function, and $(\cdot)_{n}$ is the Pochhammer symbol, which is defined as $(x)_{n}=\Gamma(x+n)/\Gamma(x)$.  The 2-norm of $\bm{x}=[x_1,x_2,\ldots,x_n]^{\rm T}$ is defined as $\left\|\bm{x}\right\|_2=\sqrt{x_1^2+\cdots+x_n^2}$. $\binom n k$ denotes the binomial coefficient. 

\section{System Model}\label{S2}

As illustrated in Fig.~\ref{fig:1}, the system under consideration is a near-shore maritime wireless network, composed of a single onshore station (OS) located at the origin and a population of randomly distributed vessels. Our performance analysis focuses on the uplink communication from an arbitrary vessel to the OS. This scenario is fundamental for supporting critical maritime services, including safety applications, navigation, and data exchange.

\subsection{System Geometry and Vessel Deployment}\label{S2.1}

To capture the concentration of maritime traffic near coastlines, the spatial locations of all vessels are modeled by a NHPPP, $\Phi_{\text{S}}$, on a two-dimensional half-plane. The process intensity $\mu(d)$ is assumed to be isotropic, depending only on the horizontal distance $d$ from the OS, and is given by $\mu(d)=\mu_0f(d)$, where $\mu_0$ is the baseline intensity constant representing the average number of vessels per unit area without the spatial modulation by $f(d)$. In this work, the spatial profile $f(d)$ is modeled using the kernel of a Gamma distribution. This modeling choice is supported by statistical analyses of the Automatic Identification System (AIS) data \cite{Braca2016}, demonstrating that maritime traffic is rarely uniform but typically exhibits a unimodal spatial profile. Specifically, vessels tend to concentrate in designated shipping lanes or anchorage zones located at a specific distance from the shore, leading to a lower vessel density in coastal proximity and a peak density at the lane center. The Gamma distribution (with shape parameter $k > 1$) naturally captures this physical characteristic, as it allows the probability density to start at zero, rise to a peak at a non-zero distance, and then decay, offering a more realistic representation than uniform or exponential models. Finally, the performance is evaluated for a randomly selected, typical vessel.

\subsection{Antenna and Channel Characteristics}\label{S2.2}

Both the OS and the vessels are assumed to be equipped with omnidirectional, vertically polarized antennas. This is a standard assumption for near-shore communications, as an omnidirectional pattern ensures connectivity regardless of a vessel's orientation, while vertical polarization optimizes performance in the presence of sea-surface reflections.

\begin{figure}[!t]
	%\vspace*{-1mm}
	\begin{center}
		\includegraphics[width=1\columnwidth]{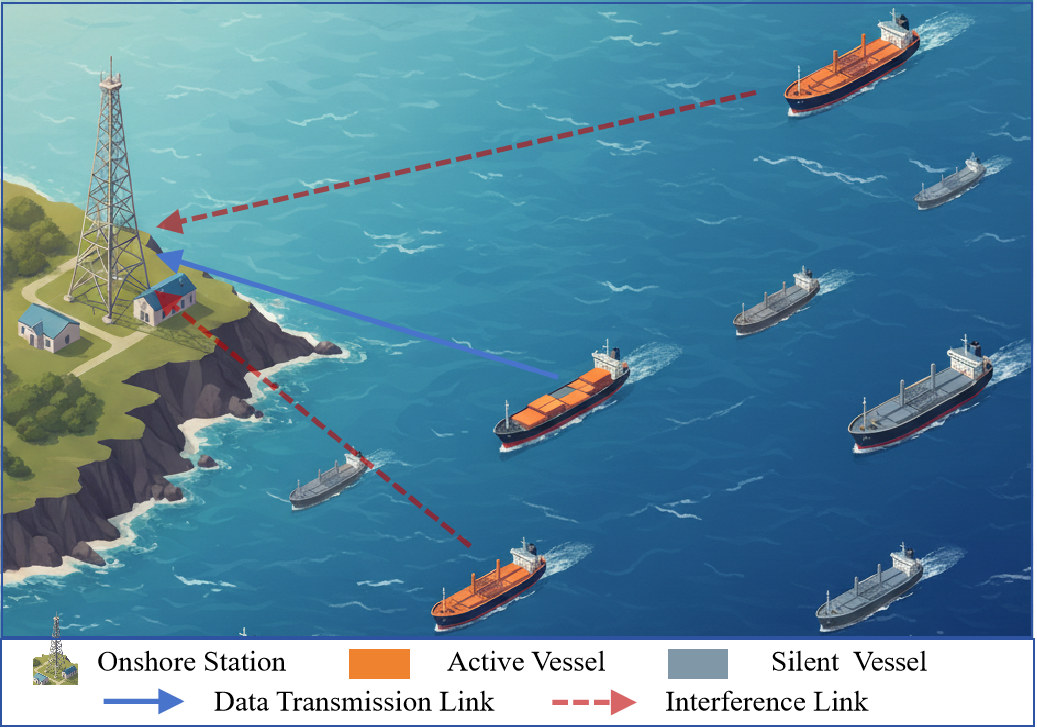}
	\end{center}
	\vspace*{-4mm}
	\caption{Illustration of the near-shore maritime wireless communication system.}
	\label{fig:1} % Fig.1
	\vspace*{-4mm}
\end{figure}

For the small-scale channel gain, $h$, we adopt the Nakagami-$m$ fading model. This model offers a strong balance between mathematical tractability and the flexibility to accurately represent the line-of-sight (LOS)-dominant conditions of the maritime channel, where scattering from the sea surface induces fading\footnote{For instance, with the Nakagami fading parameter set to $m=1$, it reduces to the Rayleigh distribution \cite{IoTJDong}. For the LOS-dominant maritime channel, it can be closely matched to the Rician-$K$ distribution by setting  $m=\frac{(K+1)^2}{2K+1}$. Furthermore, by tuning the parameter $m$, it is possible to model fading conditions from severe (e.g., rough sea) to moderate (e.g., calm sea) by fitting the distribution to empirical data.}. The channel gain power, $|h|^2$, thus follows a normalized Gamma distribution. 
	Beyond its statistical distribution, characterizing the temporal evolution of the fading channel is equally important for actual data transmission. We adopt a quasi-static assumption, where channel parameters are modeled as constant over a packet transmission interval. This assumption is substantiated by an analysis of the Doppler spread induced by wave motion. For the very high frequency (VHF) system under consideration, let $f_c = 162$ MHz denote the carrier frequency and $\lambda \approx 1.85$ m be the corresponding wavelength. Even under rough sea conditions, characterized by a significant wave height (SWH) $H_{\text{s}} > 2$ m, the orbital velocity of surface wave particles, denoted by $v_o$, remains moderate at approximately $1 \sim 2$ m/s \cite{WangAceess}. This motion induces a maximum Doppler shift, $f_d = v_o/\lambda$, of approximately $1$ Hz. The resulting channel coherence time, estimated as $T_c \approx 9/(16\pi f_d) \approx 180$ ms, significantly exceeds the $26.6$ ms transmission duration of a standard Automatic Identification System (AIS) packet. Consequently, the maritime channel can be accurately modeled as block-fading, a premise further supported by measurement campaigns reporting narrow Doppler spreads in such environments  \cite{Yang2010}.

\subsection{Baseline Medium-Access and Interference Model}\label{S2.3}

Let $p_a\in(0,1]$ denote the effective activity factor of a vessel on the considered time-frequency resource. To obtain a tractable physical-layer interference baseline, the parent vessel process $\Phi_{\mathrm{S}}$ is independently thinned with probability $p_a$. The resulting process $\Phi_{\mathrm{S}}'$ preserves the first-order active-transmitter intensity but does not reproduce the slot selection, reservation, sensing, and spatio-temporal dependence introduced by operational AIS access schemes. According to the thinning theorem \cite{1995Stochastic}, the intensity of the active-transmitter process is
\begin{equation}\label{eq2-1}
	\mu'(d)=p_a\mu(d)=p_a\mu_0f(d).
\end{equation}

\begin{remark}
	Operational AIS uses TDMA-based access schemes, including self-organizing TDMA and carrier-sense TDMA. SOTDMA introduces slot selection and reservation, whereas CSTDMA relies on channel sensing; these mechanisms generate spatio-temporal dependence among simultaneous transmitters \cite{ITURM1371}. Independent thinning captures only the marginal active-transmitter intensity $p_a\mu(d)$ and omits these protocol-induced correlations. The resulting coverage and rate expressions are therefore interpreted as conservative physical-layer benchmarks under an uncoordinated activity model, rather than as protocol-level predictions for operational AIS networks. Here, ``conservative'' means that the coordination gains provided by operational access schemes are not represented; no strict lower-bound property is claimed for every AIS deployment. A protocol-faithful extension would require a dependent space-time access process.
\end{remark}

\subsection{Large-Scale Path Loss Model}\label{S2.4}

For the near-shore propagation channel, we adopt the location-dependent, piecewise path loss framework proposed in~\cite{Lee20214}, as illustrated in Fig. \ref{fig:1.1}. This approach combines the accuracy of the two-ray model at shorter distances with the three-ray model's ability to account for atmospheric ducting at longer ranges. The transition between the models occurs at a breakpoint distance, $d_{\text{break}}$, defined as:
\begin{equation}\label{eq2-2} % eq.2
	d_{\text{break}} = \frac{4 h_{\text{t}} h_{\text{r}}}{\lambda},
\end{equation}
where $h_{\text{t}}$ and $h_{\text{r}}$ are the transmitter and receiver antenna heights above sea level, respectively, and $\lambda$ is the carrier wavelength. This framework is based on a simplified flat-earth geometry, assuming a perfectly smooth sea surface with a reflection coefficient of $\Gamma \approx -1$ and far-field conditions defined by $d \gg h_{\text{t}}, h_{\text{r}}$.

\subsubsection{Two-Ray Model ($d < d_{\rm{break}}$)}
At distances shorter than the breakpoint, the coherent sum of the direct and sea-reflected paths is dominant. The path loss in this region is described by the ideal two-ray model, given in a linear scale as:
\begin{equation}\label{eq2-3} % eq.3
	PL_{\text{2-ray}}^{\text{ideal}}(d) = \frac{(4\pi d / \lambda)^2}{\left(2\sin\left(\frac{2\pi h_{\text{t}} h_{\text{r}}}{\lambda d}\right)\right)^2}.
\end{equation}

\begin{figure}[!t]
	\vspace*{-1mm}
	\begin{center}
		\includegraphics[width=0.98\columnwidth]{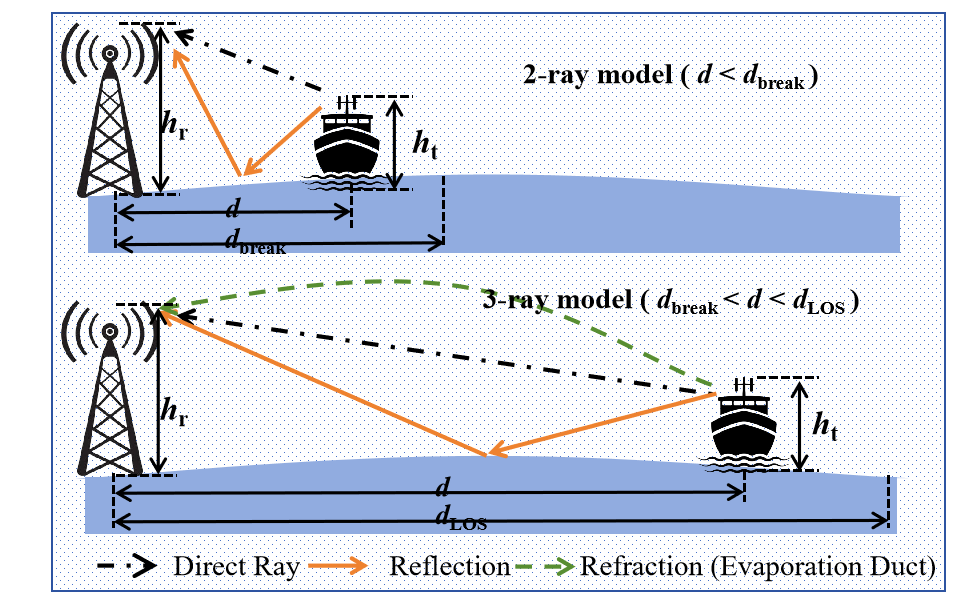}
	\end{center}
	\vspace*{-5mm}
	\caption{Illustration of the piecewise path loss model.}
	\label{fig:1.1} % Fig.2
	\vspace*{-5mm}
\end{figure}

\subsubsection{Three-Ray Model ($d_{\rm{break}} \le d < d_{\rm{LOS}}$)}
Beyond $d_{\text{break}}$, a third path component resulting from reflection and refraction within the atmospheric evaporation duct becomes significant. The path loss is then described by an ideal three-ray model up to the LOS distance, $d_{\text{LOS}}$. The linear path loss is given by:
\begin{equation}\label{eq2-4} % eq.4
	PL_{\text{3-ray}}^{\text{ideal}}(d) = \frac{(4\pi d / \lambda)^2}{\left(2\left(1+\Delta(d)\right)\right)^2},
\end{equation}
where the interference term $\Delta(d)$ is given by:
\begin{equation}\label{eq2-5} % eq.5
	\Delta(d) = 2 \sin\left(\frac{2\pi h_{\text{t}} h_{\text{r}}}{\lambda d}\right) \sin\left(\frac{2\pi (h_{\text{e}} - h_{\text{t}})(h_{\text{e}} - h_{\text{r}})}{\lambda d}\right).
\end{equation}
Here, $h_{\text{e}}$ is the effective height of the evaporation duct, and $d_{\text{LOS}} \approx \sqrt{2R_{\text{e}}h_{\text{t}}} + \sqrt{2R_{\text{e}}h_{\text{r}}}$, with $R_{\text{e}}$ being the effective Earth's radius.

\subsection{SINR Model}\label{S2.5}

Our analysis focuses on the uplink, where a vessel $S_m$ at distance $D_{S_m}$ transmits to the OS. This link is subject to co-channel interference from the set of other concurrently transmitting vessels, $\Phi_{\text{S}}'$. The performance is therefore determined by the signal-to-interference-plus-noise ratio (SINR) at the OS, formulated as:
\begin{align} % eq.6
	\mathrm{SINR}_{S_m} &= \frac{P_{S_m} |h_{S_m}|^2 G(D_{S_m})}{I_{S_m} + N_0} \label{eqSINR0}, 
\end{align}
where
\begin{align} % eq.7
 \quad I_{S_m} &= \sum_{S_i \in \Phi_{\text{S}}'} P_{S_i} |h_{i}|^2 G(D_i) \label{eqInterf}.
\end{align}
In the above, $P_{S_m}$ and $P_{S_i}$ are the transmit powers of the desired vessel $S_m$ and an interfering vessel $S_i$, respectively, with $|h_{S_m}|^2$ and $|h_i|^2$ being their respective small-scale fading power gains, while $N_0$ is the additive white Gaussian noise (AWGN) power, and the term $G(d) = 1/PL(d)$ represents the large-scale path gain at distance $d$, where $PL(d)$ is the piecewise path loss model defined by \eqref{eq2-3} or \eqref{eq2-4}.

\section{Refined Large-Scale Path Loss Model with Sea Surface Roughness}\label{S3}

Existing maritime path loss frameworks are predicated on the idealization of a perfectly smooth sea surface, leading to an ideal reflection coefficient of $\Gamma \approx -1$. This simplification is a primary source of deviation between theoretical predictions and measured data, as it fails to capture the significant impact of sea-surface roughness. To address this limitation, we introduce a more physically grounded path loss model where signal propagation is a direct function of the sea state.

\subsection{Quantifying Sea Roughness}\label{S3.1}

The key to extending the model's applicability beyond the calm or near-calm sea states considered in previous validation studies is to quantitatively incorporate surface roughness. We achieve this by using the significant wave height, denoted as $H_{\text{s}}$, a standard oceanographic parameter defined as the average height of the highest one-third of waves. 
The selection of $H_{\text{s}}$ is strategic, as it directly links our physical reflection model to the standardized World Meteorological Organization's sea state codes \cite{MoC} used in global maritime and meteorological reporting. This linkage ensures that our model's key input is readily available from numerous practical sources, including real-time buoy measurements, historical oceanographic databases, and standard meteorological forecasts.

{\color{black}
	By incorporating $H_{\text{s}}$, the proposed formulation provides a parametric means of examining how the attenuation of the coherent specular reflection component changes with sea state.
	However, this capability should not be interpreted as empirical validation across the full spectrum of operational sea states.
	In particular, the VHF measurements considered in this work correspond to calm-sea conditions, while the VHF behavior for $H_{\text{s}}>1$\,m is evaluated through model-based sensitivity analysis.
	Direct VHF measurement validation under rough sea conditions remains outside the empirical scope of the present study.
}

\subsection{Roughness-Dependent Reflection Coefficient}\label{S3.2}

The ideal reflection coefficient, $\Gamma \approx -1$, is a simplification valid only for a perfectly smooth, mirror-like sea surface. In reality, a rough sea surface fundamentally alters the reflection mechanism. The incident energy is partitioned into two distinct components: a specular component, which reflects coherently in a single, predictable direction as if from a mirror, and a diffuse component, which scatters non-coherently in many directions.
Our path loss model is predicated on the interference between coherent wave components (direct, reflected, and refracted paths). Therefore, the deterministic propagation model focuses on the coherent direct, reflected, and refracted components that produce the predictable large-scale interference structure. The unresolved diffuse field is represented only approximately through statistical small-scale fading and is not explicitly characterized in terms of its power, angular spectrum, delay profile, or spatial correlation.
The effective reflection coefficient, $\rho_{\text{s}}$, is defined to represent the amplitude and phase of this specular component. 

We employ the classical Rayleigh criterion for roughness \cite{Maradudin1976}. This criterion quantifies the degree of surface roughness as perceived by an incident wave, based on the phase variations it induces across the reflecting surface. These phase variations are captured by the dimensionless Rayleigh roughness factor $\xi$, defined as \cite{Pinel2010}:
\begin{equation}\label{eqVII-1} % eq.8
	\xi = \frac{4\pi H_{\text{s}} \sin\psi}{\lambda} ,
\end{equation}
where $H_{\text{s}}$ is the significant wave height, and $\psi$ is the grazing angle,  which for long-range links can be approximated as $\psi \approx (h_{\text{t}} + h_{\text{r}}) / d$. The factor $\xi$ represents the standard deviation of the phase difference between rays reflected from different heights on the rough surface.
	
The attenuation of the specular reflection component originates from random height variations on the rough surface, which introduce random phase shifts into the reflected wave. To mathematically model this physical process, we adopt the widely accepted assumption that the surface height variations follow a Gaussian distribution. The coherent specular field is the statistical average of the total reflected field over all random surface heights. This averaging is mathematically equivalent to finding the expected value of the random phase term, $\exp (j \delta \phi_r)$. In this term, $j$ is the imaginary unit, and $\delta \phi_r$ denotes the phase  variation of the reflected  wave. For a Gaussian process, this expectation has a standard solution which yields a negative exponential function.
This attenuation, resulting from the averaging process, is captured by the specular scattering factor $\rho_{\textrm{spec}}$, which acts as a reduction coefficient applied to the ideal reflection coefficient \cite{Pinel2010}:
\begin{equation} % eq.9
	\rho_{\text{spec}} = \exp\left(-\frac{\xi^2}{2}\right). \label{eq:rho_spec}
\end{equation}
This formulation correctly models the physical behavior:
\begin{itemize}
	\item When the sea is smooth, i.e., $H_{\mathrm{s}}\to0$, then $\xi\to0$ and $\rho_{\mathrm{spec}}\to1$, so that the coherent specular reflection component approaches the smooth-sea limit.
	\item As $H_{\mathrm{s}}$ increases, $\xi$ increases and $\rho_{\mathrm{spec}}$ decreases, indicating attenuation of the ensemble-averaged coherent specular reflection component. The present formulation does not explicitly determine the power or angular distribution of the corresponding diffuse field.
\end{itemize}

\textit{Electrical Roughness, Frequency Scaling, and Validation Scope:} The Rayleigh roughness parameter $\xi$ in \eqref{eqVII-1} contains the carrier wavelength $\lambda$ in the denominator. Hence, the same physical sea state $H_{\text{s}}$ produces different electrical roughness at different carrier frequencies. A shorter wavelength produces stronger attenuation of the coherent specular reflection component, whereas a longer wavelength produces weaker attenuation. In the electrically smooth limit, $\xi\to0$, the effective reflection coefficient approaches the ideal smooth-sea value $\rho_{\text{s}}\to-1$, and Lee's ideal model is recovered. Accordingly, the roughness factor is evaluated independently using the wavelength of each operating band; no factor calculated at 5.15~GHz is used in the 162~MHz propagation or network calculations. The 5.15~GHz measurements assess the roughness-sensitive attenuation mechanism at C-band but cannot substitute for direct rough-sea validation at VHF. Conversely, the 160~MHz calm-sea measurements examine the electrically smooth limiting behavior near the AIS frequency but do not validate the predicted VHF behavior for $H_{\text{s}}>1$~m.

\textit{Applicability under Severe Sea States:} The Rayleigh factor adopted in this work is a coherent-field reduction factor. It characterizes attenuation of the ensemble-averaged specular component but does not determine the power, angular distribution, delay structure, or spatial correlation of the diffuse field generated by a severely rough sea surface. In particular, the limit $\rho_{\mathrm{spec}}\to0$ indicates suppression of the coherent specular reflection component; it does not imply that the scattered energy vanishes or that the remaining field is fully described by the deterministic two-ray/three-ray model. The Nakagami-$m$ term used in the network analysis represents unresolved small-scale fluctuations phenomenologically. Because its parameter is not derived from a sea-state-dependent diffuse-scattering model, it should not be interpreted as a physical substitute for the missing diffuse field. Accordingly, the quantitative sea-state interpretation of the present model is limited to $H_{\mathrm{s}}\leq3$~m. Results for $H_{\mathrm{s}}>3$~m are retained only as parametric stress tests of coherent-reflection suppression and are not used as operational severe-sea propagation or network predictions. The value $H_{\mathrm{s}}=3$~m is adopted as a conservative reporting boundary in this work rather than as a universal electromagnetic transition, which also depends on wavelength, grazing geometry, and the sea-surface spectrum.

The modeling of the effective reflection from a rough surface is decomposed into two distinct physical effects. The first is the intrinsic electromagnetic interaction at the air-sea interface, described by the ideal, smooth-surface reflection coefficient  $\Gamma$. The second effect is the attenuation of this ideal field by the surface's geometric roughness. This roughness causes a fraction of the energy to be incoherently scattered away from the specular direction, an effect quantified by the dimensionless specular scattering factor $\rho_{\text{spec}}$, which represents the fraction of the field that remains coherent. The overall effective reflection coefficient, $\rho_{\text{s}}$, is the result of these two sequential effects, and can be modeled by their product. 
Thus, the ideal reflection coefficient is replaced with an effective reflection coefficient $\rho_{\text{s}}$, which depends on $H_{\text{s}}$ and is derived from the classical Rayleigh criterion for roughness \cite{Davies1954}:
\begin{equation}\label{eqERC} % eq.10
	\rho_{\text{s}} = \Gamma \cdot \rho_{\text{spec}} \approx -1 \cdot \exp\left(-\frac{\xi^2}{2}\right).
\end{equation}

\noindent\textit{Limitation under Severe Sea States:} The Rayleigh roughness factor adopted in this work characterizes attenuation of the ensemble-averaged coherent specular reflection component. It does not explicitly determine the power, angular distribution, delay structure, or spatial correlation of the diffuse field generated by a severely rough sea surface. In particular, a reduction in $\rho_{\mathrm{spec}}$ represents progressive suppression of the coherent specular reflection component but does not provide a complete description of the remaining scattered field. The Nakagami-$m$ fading term represents unresolved small-scale fluctuations statistically; however, because its parameter is not derived from a sea-state-dependent diffuse-scattering model, this treatment constitutes a coarse approximation rather than an explicit physical representation of diffuse scattering. This limitation becomes increasingly important under severe sea states, particularly for $H_{\mathrm{s}}>3$~m, where diffuse scattering may dominate. Results in this regime should therefore be interpreted as sensitivity predictions under the adopted coherent-reflection abstraction rather than as complete quantitative descriptions of severe-sea propagation.

\subsection{Modified Path Loss Equations}\label{S3.3}

The introduction of $\rho_{\text{s}}$ necessitates a modification of the original path loss equations to account for non-ideal reflection.
\begin{Theorem}\label{T1}
Under the assumption of a non-ideal sea surface reflection characterized by the effective reflection coefficient $\rho_{\mathrm{s}}$, the linear path loss, $P\!L^{\rm{rough}}(d)$, for the near-shore maritime channel is given by a piecewise function.
For $d < d_{\mathrm{break}}$ (Two-Ray Region), it is expressed by
\begin{equation} \label{eq3-4} % eq.11
	P\!L_{\rm{2}\textrm{-}\rm{ray}}^{\rm{rough}}(d) = \frac{(4\pi d / \lambda)^2}{G_{\mathrm{mp}}^{\rm{2}\textrm{-}\rm{ray}}},
	\end{equation}
and for $d \geq d_{\mathrm{break}}$ (Three-Ray Region), it is given by
\begin{equation}\label{eq3-5} % eq.12
	P\!L_{\rm{3}\textrm{-}\rm{ray}}^{\mathrm{rough}}(d) = \frac{(4\pi d / \lambda)^2}{G_{\mathrm{mp}}^{\rm{3}\textrm{-}\rm{ray}}},
\end{equation}
where 
\begin{align} % eqs.13,14
	G_{\mathrm{mp}}^{\rm{2}\textrm{-}\rm{ray}} =& 1 + \rho_{\mathrm{s}}^2 - 2|\rho_{\mathrm{s}}|\cos\left(\frac{4\pi h_{\mathrm{t}} h_{\mathrm{r}}}{\lambda d}\right), \label{eqGtwoR} \\
  G_{\mathrm{mp}}^{\rm{3}\textrm{-}\rm{ray}} =& (1 + \rho_{\mathrm{s}} C_2 - C_3)^2 + (\rho_{\mathrm{s}} S_2 - S_3)^2, \label{eqGthreeR}
\end{align}  
with $C_2\! =\! \cos(k\Delta d_2)$, $S_2\! =\! \sin(k\Delta d_2)$, $C_3\! =\! \cos(k\Delta d_3)$,  $S_3\! =\! \sin(k\Delta d_3)$, $\Delta d\! =\! \Delta d_2\! =\! \frac{2 h_{\mathrm{t}} h_{\mathrm{r}}}{d}$ and $\Delta d_3\! =\! \frac{2(h_{\mathrm{e}} - h_{\mathrm{t}})(h_{\mathrm{e}} - h_{\mathrm{r}})}{d}$.
\end{Theorem}
\begin{proof}
See Appendix~\ref{ApA}.
\end{proof}

\begin{remark}
Note that $\Delta d$, $\Delta d_{2}$ and $\Delta d_{3}$ denote the geometric path length difference in meters, whereas $\Delta(d)$ in (5) denotes the dimensionless interference term used in the ideal three-ray formulation.
\end{remark}

This refinement transforms the path loss model from a static representation into a dynamic one capable of predicting propagation characteristics under various sea states.

\subsection{Validation Against Classical Models}\label{S3.4}

To validate the proposed generalized path loss model, we demonstrate its reduction to the classical two-ray and three-ray models under the ideal condition of a perfectly smooth sea surface, for which the reflection coefficient is $\rho_{\text{s}} = -1$.

\subsubsection{Two-Ray Model}
Under the ideal condition, our generalized multipath gain factor $G_{\text{mp}}$ reduces to:
\begin{align}\label{eqGtwoR-red} % eq.15
	\left. G_{\text{mp}}^{\textrm{2-ray}} \right|_{\rho_{\text{s}}=-1} = \left(2\sin\left(\frac{2\pi h_{\text{t}} h_{\text{r}}}{\lambda d}\right)\right)^2.
\end{align}
The resulting path loss expression becomes identical to the classical two-ray model in \cite{Lee20214}, confirming our framework is a robust generalization of the established theory.

\subsubsection{Three-Ray Model}
Unlike the simplified, engineering-oriented three-ray model given in \cite{Lee20214}, our model is rigorously derived from the first principles of wave superposition. For the ideal case of $\rho_{\text{s}}\! =\! -1$, our generalized gain factor becomes:
\begin{align}\label{eqGthreeR-red} % eq.16
	& \left. G_{\text{mp}}^{\text{3-ray}} \right|_{\rho_{\text{s}}=-1} =
	\left.(1 + \rho_{\text{s}} C_2 - C_3)^2 + (\rho_{\text{s}} S_2 - S_3)^2\right|_{\rho_{\text{s}}=-1} \nonumber\\
	& = 2\! -\! 2\cos(k\Delta d_2)\! -\! 2\cos(k\Delta d_3)\! +\! 2\cos(k(\Delta d_2\! -\! \Delta d_3)).
\end{align}
While not algebraically identical to the simplified form, this result provides a more physically accurate baseline for the ideal three-ray interference mechanism. The strength of our model lies in its ability to seamlessly extend from this rigorous baseline to realistic, non-ideal environments by varying $\rho_{\text{s}}$.

\section{Distribution of Distances}\label{S4}

The spatial distribution of vessels is modeled as a 2D NHPPP. We assume the process is radially symmetric with respect to the observer at the origin, meaning the intensity function $\mu$ depends only on the Euclidean distance $d = \|\mathbf{x}\|$ from the origin. This intensity is formulated as:
\begin{equation}\label{eqNHPPP-sym} % eq.17
	\mu(d) = \mu_0 f(d),
\end{equation}
where $\mu_0$ is the baseline intensity constant and $f(d)$ is a dimensionless function characterizing the spatial density profile.
For a point process with this radial symmetry, the probability density function (PDF) of the distance $d$ for a typical vessel, denoted $g(d)$, is given by:
\begin{align}\label{eq:g_d_general} % eq.18
	g(d) = \frac{d \cdot f(d)}{\int_{0}^{\infty} \rho f(\rho) \,\mathrm{d}\rho},
\end{align}
where the integral in the denominator is a normalization constant ensuring that $g(d)$ integrates to unity. The corresponding cumulative distribution function (CDF) is then:
\begin{equation}\label{NHPPP-eqCDF} % eq.19
	G(d) = \frac{\int_{0}^{d} \rho f(\rho) \,\mathrm{d}\rho}{\int_{0}^{\infty} \rho f(\rho) \,\mathrm{d} \rho}.
\end{equation}

To model realistic maritime scenarios where vessel density is not highest at the coast but peaks at a certain distance offshore (e.g., in shipping lanes or anchorage zones), we adopt a spatial profile $f(d)$ based on the Gamma distribution. This form is chosen for its flexibility in capturing such rise-and-fall density profiles while remaining mathematically tractable. We define $f(d)$ using the kernel of a Gamma PDF:
\begin{equation}\label{eq:f_gamma} % eq.20
	f(d) = d^{\alpha-1} \exp ({-\beta d}) ,
\end{equation}
where $\alpha \ge 1$ and $\beta > 0$ are the shape and rate parameters, respectively.
Substituting \eqref{eq:f_gamma} into the general PDF expression \eqref{eq:g_d_general} yields the final distance distribution for our model. The integral in the denominator can be solved using the definition of the Gamma function, resulting in:
\begin{align} \label{eqgd} % eq.21
	g(d) &= \frac{\beta^{\alpha+1} d^\alpha \exp ({-\beta d})}{\Gamma(\alpha\! +\! 1)} , ~ d \ge 0.
\end{align}
This resulting PDF is itself a Gamma distribution with a new shape parameter $\alpha' = \alpha+1$ and rate parameter $\beta' = \beta$.

\section{Performance Analysis}\label{S5}

In this section, we analyze the system performance in terms of coverage probability. The coverage probability is defined as the likelihood that the SINR at a receiver exceeds a predefined threshold, required for successful signal demodulation. A user is considered to be in coverage if this condition is met.

\subsection{Coverage Probability}\label{S5.1}

The main challenge in the uplink arises from co-channel interference generated by other transmitting vessels. To analyze its impact on coverage, we first derive the Laplace transform of the aggregate interference.

\begin{Lemma}\label{L1}
The Laplace transform of the aggregate interference, $\mathcal{L}_{I_{S_m}}(s)$, is derived by considering the contributions from interferers in the near-field (up to $d_{\rm{break}}$) and far-field (beyond $d_{\rm{break}}$), corresponding to the two-ray model:
\begin{align}\label{eqL1near} % eq.22
	\mathcal{L}_{I, \mathrm{near}}(s) =& \exp \Bigg(\pi \int_0^{d_{\mathrm{break}}} \left( {\left( \!1\!+\!\frac{s P_{S_i} }{ m  	P\!L_{\rm{2}\textrm{-}\rm{ray}}(d_i) }\!\right)^{\!\!\!-m}}-1 \right) \nonumber \\
	& \times p_a \mu_0  d_i^{\alpha} \exp ({-\beta d_i}) \,\mathrm{d}d_i \Bigg),
\end{align}
and the three-ray path loss model:
\begin{align}\label{eqL1far} % eq.23
	\mathcal{L}_{I, \mathrm{far}}(s_1) = &\exp \Bigg(\pi \int_{d_{\mathrm{break}}}^{d_{\mathrm{LOS}}} \left( {\left(\!1\!+\!\frac{s_1 P_{S_i}}{  m  P\!L_{\rm{3}\textrm{-}\rm{ray}}(d_i) }\!\right)^{\!\!\!-m}}-1 \right) \nonumber 
\end{align}
\begin{align}
	& \times p_a \mu_0  d_i^{\alpha} \exp ({-\beta d_i}) \,\mathrm{d}d_i \Bigg),
\end{align}
respectively, where $p_a$ denotes the effective activity factor on the considered time-frequency resource.
\end{Lemma}

\begin{proof}
See Appendix~\ref{ApB}.
\end{proof}

\begin{Theorem}\label{The2}
The uplink coverage probability for a typical vessel  connecting to the OS, averaged over its spatial distribution $g(d)$, is given by:
\begin{align}\label{eqThe21} % eq.24
	P_{\mathrm{cov}}^{\mathrm{UL}} =& \int_0^{d_{\rm{break}}} \,\sum_{n=1}^{m}(-1)^{n+1}\binom{m}{n}\exp\left(-s{N_0} \right) \nonumber \\
	& \hspace*{-5mm} \times\exp \Bigg(\pi \int_0^{d_{\mathrm{break}}} \left( {\left( \!1\!+\!\frac{s P_{S_i} }{ m P\!L_{\rm{2}\textrm{-}\rm{ray}}(d_i)  }\!\right)^{\!\!\!-m}}\!\!\! -1 \right) \nonumber \\
	& \hspace*{-5mm} \times p_a \mu_0  d_i^{\alpha} \exp ({-\beta d_i}) \mathrm{d}d_i \Bigg) \frac{\beta^{\alpha+1} d^\alpha \exp({-\beta d})}{\Gamma(\alpha+1)} \,\mathrm{d}d \nonumber \\
	& \hspace*{-5mm} + \int_{d_{\mathrm{break}}}^{d_{\rm{LOS}}} \sum_{n=1}^{m}(-1)^{n+1} \binom{m}{n} \exp \left( -s_1 {N_0}  \right) \nonumber \\
	& \hspace*{-5mm} \times \exp \Bigg(\pi \int_{d_{\mathrm{break}}}^{d_{\mathrm{LOS}}} \left( {\left(\!1\!+\!\frac{s_1 P_{S_i} }{   m P\!L_{\rm{3}\textrm{-}\rm{ray}}(d_i)  }\!\right)^{\!\!\!-m}}\!\!-1 \right) \nonumber \\
	& \hspace*{-5mm}\times p_a \mu_0  d_i^{\alpha} \exp ({-\beta d_i}) \,\mathrm{d}d_i\! \Bigg) \frac{\beta^{\alpha+1} d^\alpha \exp ({-\beta d})}{\Gamma(\alpha+1)} \,\mathrm{d}d,
\end{align}
where $s\! =\! \frac{-n \,\eta \,T\,P\!L_{\rm{2}\textrm{-}\rm{ray}}(d)}{P_{S_m}}$, $s_1\! =\! \frac{-n \,\eta \,T\,P\!L_{\rm{3}\textrm{-}\rm{ray}}(d)}{P_{S_m}}$, $P\!L$ is either $P\!L^{\rm{ideal}}$, given by (\ref{eq2-3}) and (\ref{eq2-4}), or $P\!L^{\rm{rough}}$, given by (\ref{eq3-4}) and (\ref{eq3-5}).
\end{Theorem}

\begin{proof}
See Appendix~\ref{ApC}. 
\end{proof}

\subsection{Average Ergodic Rate}\label{S5.2}
	
The average ergodic rate, measured in bits per second per Hertz (bit/s/Hz), is a key performance metric that characterizes the long-term achievable spectral efficiency averaged over the fading states. Formally, the average ergodic rate is defined as:
\begin{eqnarray}\label{eqAER}	% eq.25
	\bar{C} \triangleq \frac{1}{K}\mathbb{E}\left[ \log_2\left(1+\text{SINR}\right) \right] .
\end{eqnarray}

\begin{Theorem}\label{T3}
	 Assuming Nakagami-$m$ fading, the average ergodic rate for a typical vessel in the uplink is given by:
\begin{align}\label{eqT4} % eq.26
	& \bar{C}^{\mathrm{UL}} \nonumber \\
	&= \frac{1}{K}\int_{t>0}\Bigg(\int_0^{d_{\rm{break}}}\sum_{n=1}^{m}(-1)^{n+1}\binom{m}{n} \exp\left(-s'{N_0} \right) \nonumber \\
	& \times\exp \Bigg(\pi \int_0^{d_{\rm{break}}}\left( {\left( \!1\!+\!\frac{s' P_{S_i} }{ m \, P\!L_{\rm{2}\textrm{-}\rm{ray}}(d_i)  }\!\right)^{\!\!\!-m}}-1 \right])\nonumber \\
	& \times p_a \mu_0 d_i^{\alpha} \exp ({-\beta d_i}) \,\mathrm{d}d \Bigg) \frac{\beta^{\alpha+1} d^\alpha \exp ({-\beta d})}{\Gamma(\alpha+1)} \,\mathrm{d}d \nonumber\\
		& + \int_{d_{\mathrm{break}}}^{d_{\rm{LOS}}} \sum_{n=1}^{m}(-1)^{n+1} \binom{m}{n} \exp  \left( -s_1' {N_0}  \right) \nonumber 
\end{align}
\begin{align}	
	& \times\exp \left(\pi \int_{d_{\rm{break}}}^{d_{\rm{LOS}}} \left( {\left(\!1\!+\!\frac{s_1' P_{S_i} }{   m \, P\!L_{\rm{3}\textrm{-}\rm{ray}}(d_i) }\!\right)^{\!\!\!-m}}\!\!\!\!-1 \right) \right.\nonumber \\
	& \times p_a \mu_0 d_i^{\alpha} \exp ({-\beta d_i}) \,\mathrm{d}d_i\! \Bigg) \frac{\beta^{\alpha+1}  d^\alpha \exp ({-\beta d})}{\Gamma(\alpha+1)}\mathrm{d}d \Bigg)\mathrm{d}t ,
\end{align}	
where $s'\! =\! \frac{-n \,\eta \,(2^t-1)\,P\!L_{\textrm{2-ray}}(d)}{P_{S_m}}$ and $s_1'\! =\! \frac{-n \,\eta \,(2^t-1)\,P\!L_{\textrm{3-ray}}(d)}{P_{S_m}}$.
\end{Theorem}

\begin{proof}
See Appendix~\ref{ApD}.
\end{proof}

\section{Numerical Results}\label{S6}

In this section, we verify the derived analytical expressions using Monte Carlo simulations with 100,000 runs. The results from the analytical derivations are indicated in the following figures as `Analysis', while the Monte Carlo results are indicated in the figures as `Simulation'. Unless otherwise specifically stated, the default system parameters and values utilized in the simulations are listed in Table~\ref{Table2},  where the carrier frequency is set to $f_c=162$ MHz according to the Global Maritime Distress and Safety System Manual of International Maritime Organization \cite{imo2024gmdss}. Unless otherwise stated, the network-level results use the independent-thinning baseline with $p_a=0.2$. They quantify propagation-driven coverage and rate trends under a controlled mean activity load and should not be interpreted as packet-level or protocol-level performance predictions for operational SOTDMA- or CSTDMA-based AIS networks.
\begin{table}[h] \scriptsize
	\vspace*{-1mm}
	\caption{Default Simulation System Parameters.} 
	\label{Table2}
	\vspace*{-6mm}
	\begin{center}
		\resizebox{\columnwidth}{!}{
			\begin{tabular} {c|c|c}
				\Xhline{1.2pt}
				\toprule
				\textbf{Notation} & \textbf{Parameters}& \textbf{Values} \\
				\midrule
				Vessel Transmit Power & $P_{S_m}$ & 30 dBm  \\
				Noise Power & $N_0$ & $-174$ dBm \\
				Carrier Frequency & $f_c$ & 162 MHz \\
				Bandwidth & $B$ & 25 kHz \\
				Nakagami-$m$  Parameter & $m$ & 3 \\
				BS's Receiver Antenna Height & $h_{\textrm{r}}$ & 50 m \\
				Vessel's Transmitter Antenna Height & $h_{\textrm{t}}$ & 10 m \\
				Evaporation Duct Height & $h_{\textrm{e}}$ & 30.5 m \\
				Earth Radius & $R_{\text{earth}}$ & 6371 km \\
				Annular/Gamma Shape & $\alpha$ & 3 \\
				Annular/Gamma Rate & $\beta$ & 1/500 km \\
				Vessel Base Density & $\mu_0$ & $1 \times 10^{-4}$ m$^{-2}$ \\
				Effective Activity Factor & $p_a$ & 0.2  \\
				\bottomrule
			\end{tabular}	
		}
	\end{center}
	\vspace*{-3mm}
\end{table}

The sea-state-dependent curves reported below are generated under the adopted coherent-reflection attenuation model. For severe sea states, particularly $H_{\mathrm{s}}>3$~m, the unresolved diffuse field may become dominant, whereas the present model represents it only approximately through statistical small-scale fading. Results in this regime are therefore used to illustrate the sensitivity and extrapolated behavior of the adopted model rather than to claim a complete quantitative prediction of severe-sea propagation.
\begin{figure}[!b]
	\vspace*{-5mm}
	\begin{center}
		\includegraphics[width=1\linewidth]{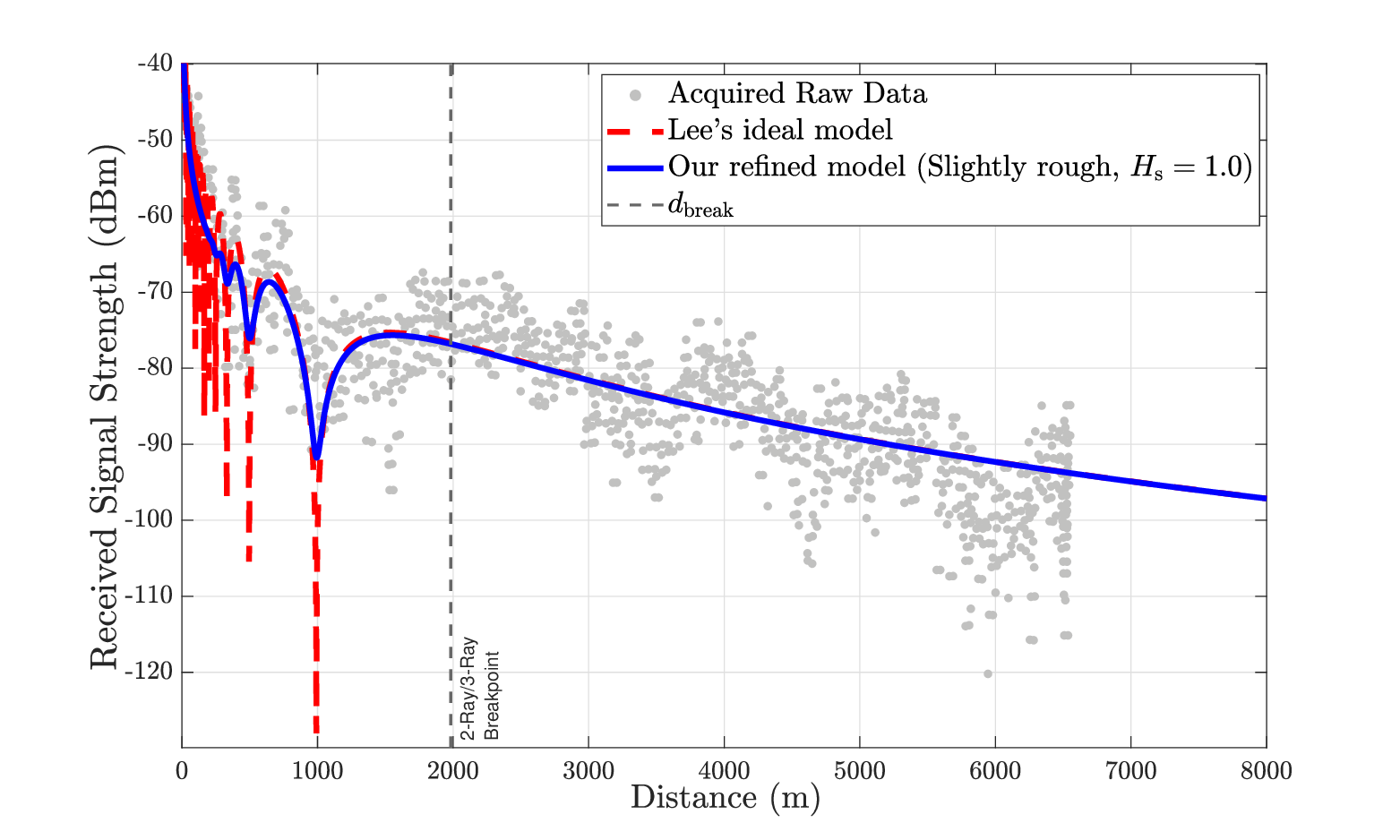}
	\end{center}
	\vspace*{-7mm}
	\caption{Comparison of Lee's ideal and our refined two-ray models with measured signal strength data.}
	\label{fig:2} % Fig.3
	\vspace*{-1mm}
\end{figure}

\subsection{Refined Path Loss Model}\label{S6.1}

Fig.~\ref{fig:2} presents a qualitative comparison between the model predictions and empirical data from Lee \cite{Lee20214} within the two-ray propagation region. 

{\color{black}
	In this subsection, we first consider the measurement campaign in \cite{Lee20214}, whose carrier frequency is $f_c=5.15$~GHz.	At this shorter wavelength, $\lambda\approx0.058$~m, the sea surface is electrically rougher than at VHF under the same physical sea state, making the attenuation of the coherent specular reflection component more visible. Accordingly, this dataset is used only as a C-band, roughness-sensitive comparison of the Rayleigh-criterion-based attenuation mechanism.	Because electrical roughness is wavelength dependent, the C-band comparison is not used as a substitute for direct rough-sea validation of the subsequent 162~MHz results.
	
	The VHF evidence is considered separately. The 160~MHz measurements examine the calm-sea, electrically smooth limit near the AIS frequency, while the 162~MHz rough-sea cases are presented as model-based sensitivity analyses.
	Thus, the two measurement datasets provide complementary but non-interchangeable evidence.
}

Both the ideal and refined models correctly identify the locations of the multipath fading nulls. 
However, the ideal model, which assumes a perfect reflection coefficient ($\rho_{\text{s}}=-1$), generates deep nulls that underestimate the measured signal strength within these fades. 
Using the representative setting $H_{\text{s}}=1$~m, selected with reference to the wind-speed range of 10--15 knots reported for the measurement campaign \cite{Lee20214}, the refined model produces shallower fading nulls and a lower residual error than the ideal smooth-sea model. This C-band comparison assesses whether attenuation of the coherent specular reflection component can account for the observed fading depths; it is not used to validate the subsequent rough-sea VHF predictions.

\begin{figure}[!b]
	\vspace*{-5mm}
	\begin{center}
		\includegraphics[width=1\linewidth]{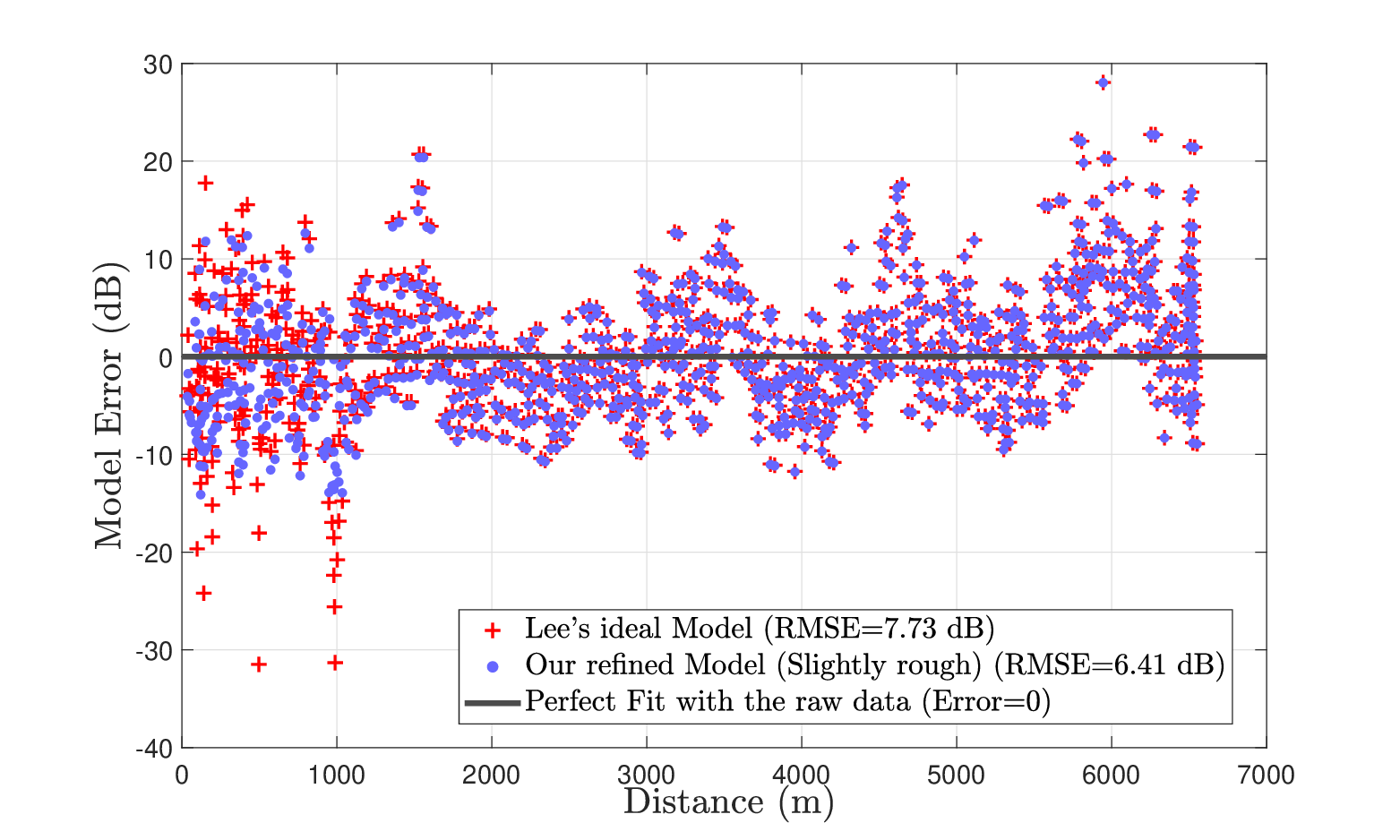}
	\end{center}
	\vspace*{-7mm}
	\caption{Residual plot and goodness-of-fit analysis for the two-ray models.}
	\label{fig:3} % Fig.4
	\vspace*{-1mm}
\end{figure}

The residual analysis is further illustrated in Fig.~\ref{fig:3}. To maintain theoretical consistency, the goodness-of-fit is evaluated exclusively for data points within the valid two-ray model range  $d < d_{\mathrm{break}}$. The residuals for the ideal model exhibit a large variance and a systematic negative bias, particularly at distances corresponding to the fading nulls. The residuals for the refined model, however, are more tightly clustered around the zero-error line, indicating a more consistent prediction.
This visual assessment is corroborated by the RMSE, which is reduced from 7.73~dB for the ideal model to 6.41~dB for the refined model within this region.
{\color{black}This reduction provides quantitative support for the roughness-induced attenuation mechanism within the considered 5.15~GHz measurement campaign. It should not be interpreted as direct validation of rough-sea path loss at 160--162~MHz.}

\begin{figure}[!t]
	\vspace*{-1mm}
	\begin{center}
		\includegraphics[width=1\linewidth]{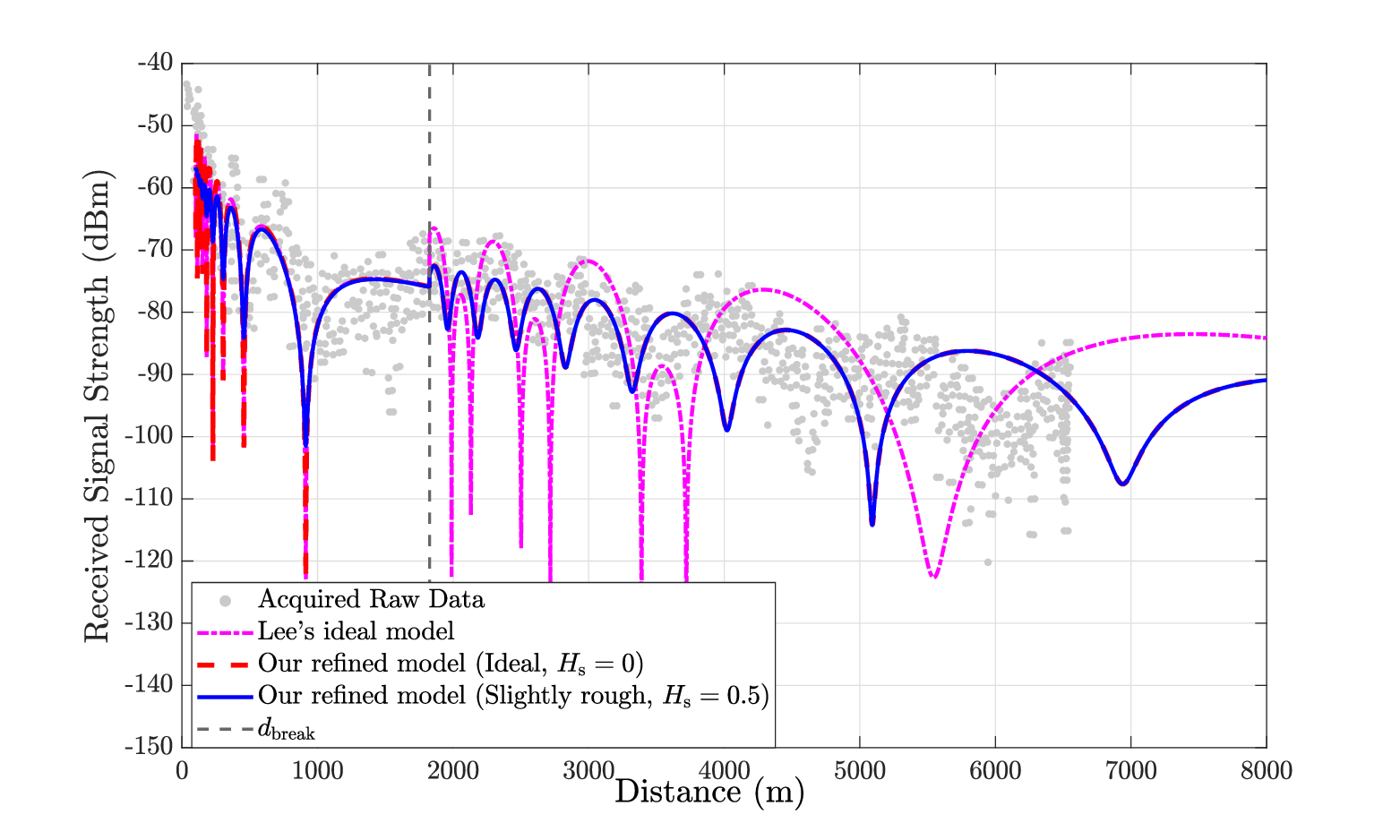}
	\end{center}
	\vspace*{-7mm}
	\caption{Comparison of Lee's ideal and our refined three-ray models with measured signal strength data.}
	\label{fig:4} % Fig.5
	\vspace*{-5mm}
\end{figure}

We next evaluate the model over the full propagation range, including the piecewise transition from the two-ray to the three-ray regime. The model fits and their corresponding residuals are presented in Figs.~\ref{fig:4} and \ref{fig:5}, respectively. 

A primary observation from Fig.~\ref{fig:4} is the difference in fading characteristics between the near-field ($d\! <\! d_{\text{break}}$) and far-field ($d\! >\! d_{\text{break}}$) regions. In the near-field, the signal exhibits frequent, deep fading nulls. This behavior is a consequence of the two-ray interference mechanism, where the phase difference between the direct and reflected paths changes rapidly at short ranges. The predictions of both the classical Lee model and the refined model are consistent with this behavior, as the influence of the evaporation duct is negligible at these distances. Beyond the break distance, the performance of the models diverges. The Lee model (dashed magenta line) continues to predict deep, periodic nulls and a rate of decay that results in an underestimation of the measured signal strength. In contrast, the refined model (solid black line), which incorporates the evaporation duct path, captures two key far-field effects: an elevation in the average signal strength and a reduction in fading depth. The refined model's predictions align more closely with the empirical data, suggesting that the three-ray approach is necessary for improved accuracy at longer ranges. 

\begin{figure}[!b]
	\vspace*{-5mm}
	\begin{center}
		\includegraphics[width=1\linewidth]{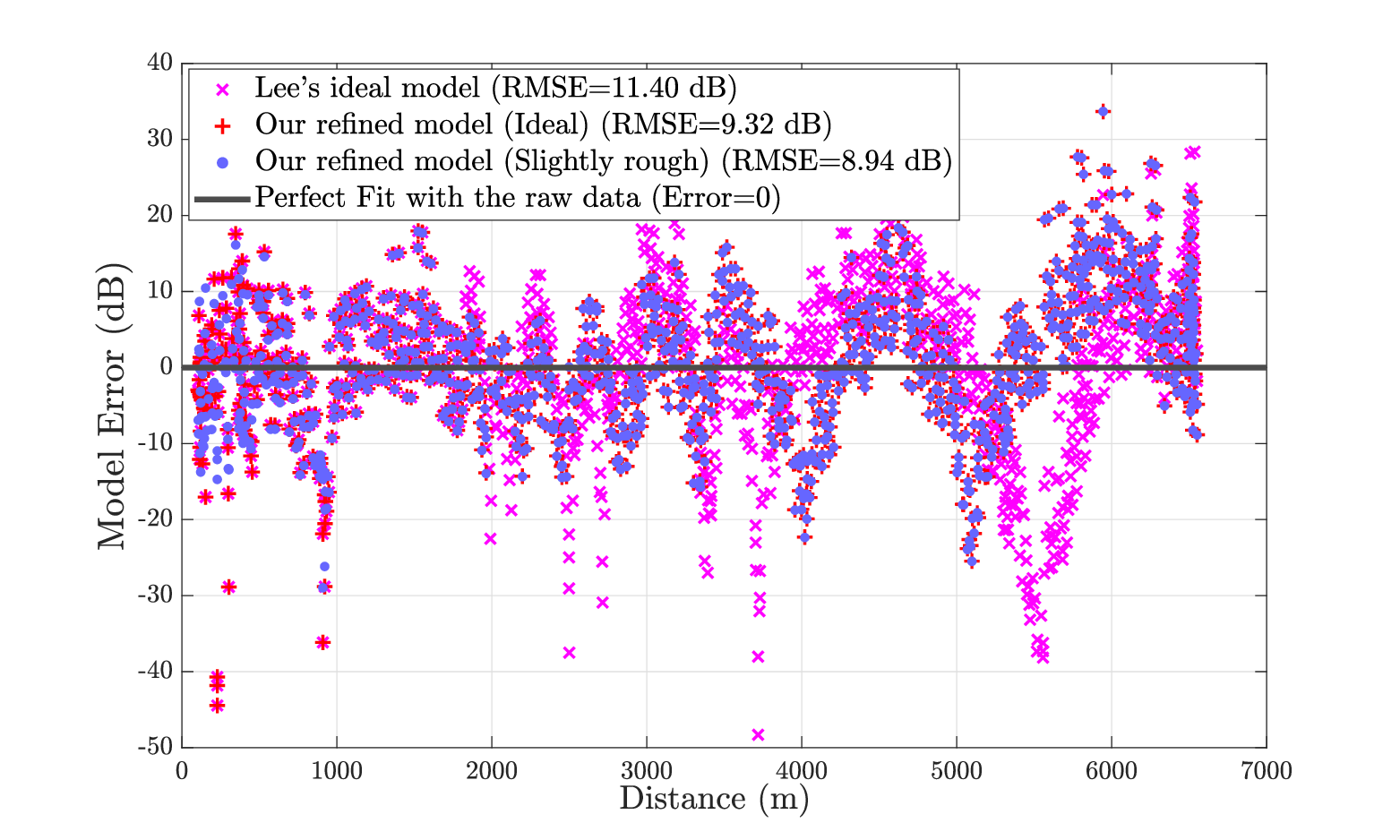}
	\end{center}
	\vspace*{-7mm}
	\caption{Residual plot and goodness-of-fit analysis for the three-ray models.}
	\label{fig:5} % Fig.6
	\vspace*{-1mm}
\end{figure}

A secondary observation pertains to the internal consistency of the refined model. In the far-field, the predictions for ideal and slightly rough conditions converge. This is consistent with the underlying physics, as the effect of surface roughness diminishes at small grazing angles $\psi$  characteristic of large distances. The roughness factor, $\xi \propto \sin(\psi)$, approaches zero as $d$ increases, causing the effective reflection coefficient to approach its ideal value of~-1. The model's behavior is consistent with the physical principle that the sea surface appears electromagnetically smoother at extended ranges.

The quantitative analysis from the residuals plot in Fig.~\ref{fig:5} reinforces these findings. While both models show large errors in the near-field, the refined model has a lower overall RMSE across the full range, i.e., 8.94~dB vs. 11.40~dB. {\color{black}
	Within this 5.15~GHz measurement campaign, the lower RMSE indicates that accounting for shallower fades associated with the attenuation of coherent reflection, together with the duct-induced component, improves agreement over the considered propagation range.
	This frequency-specific comparison is not extrapolated as empirical evidence for rough-sea VHF propagation.
}

\begin{figure}[!b]
	\vspace*{-4mm}
	\centering
	\subfigure[Calm-sea comparison with 160~MHz maritime measurements.]{
		\includegraphics[width=0.95\linewidth]{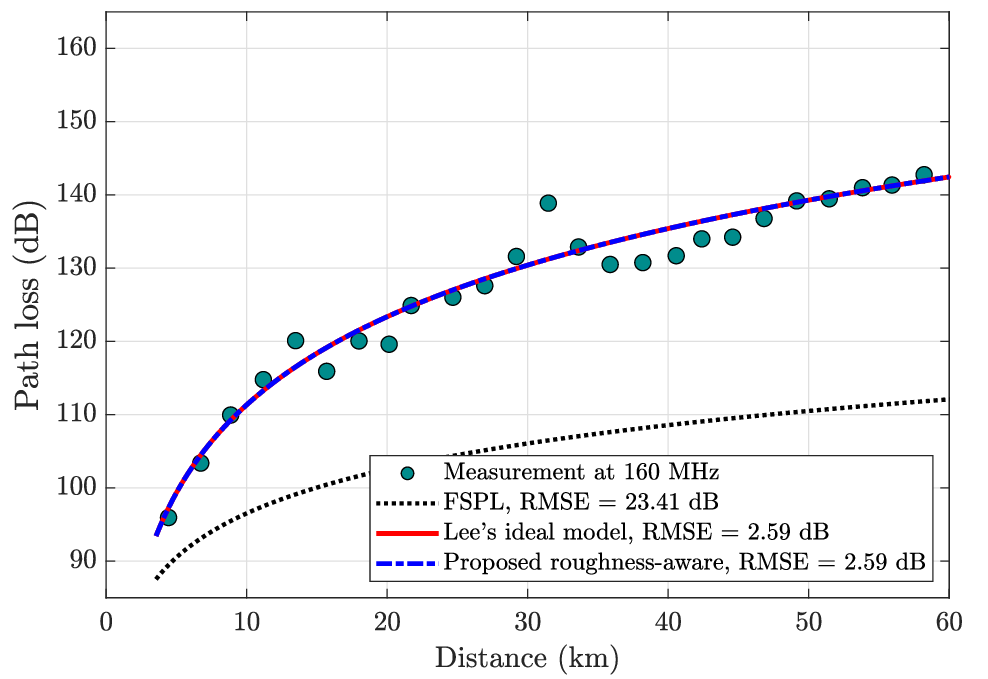}
		\label{Fig:VHF_a}}
	\subfigure[Wavelength-specific sea-state sensitivity predicted at $f_c=162$~MHz.]{
		\includegraphics[width=0.95\linewidth]{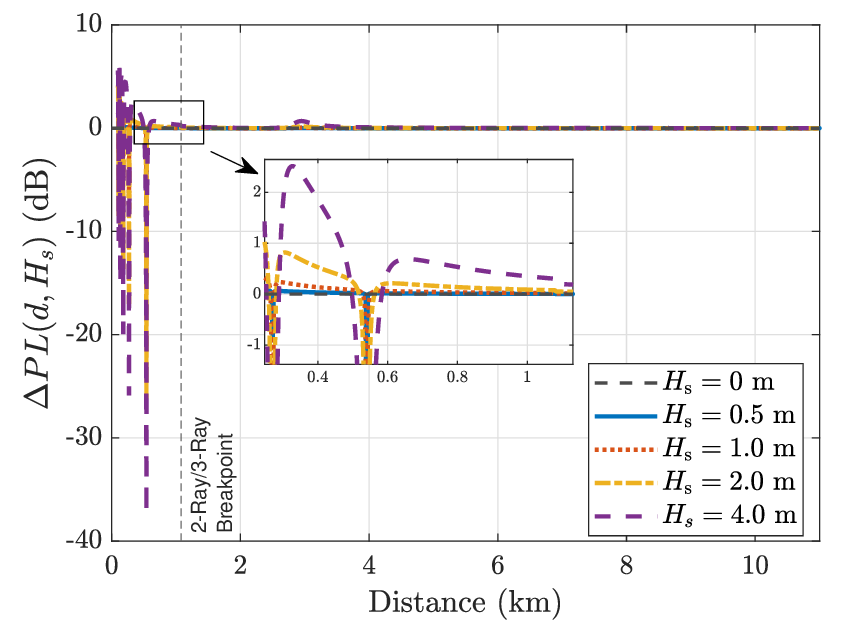}
		\label{Fig:VHF_b}}
	\vspace*{-2mm}
	\caption{Calm-sea VHF measurement comparison and wavelength-specific model sensitivity. Panel (a) examines the electrically smooth VHF limit under calm-sea conditions. Panel (b) presents model predictions and does not use rough-sea VHF measurement data.}
	\label{fig:VHF_validation}
	\vspace*{-1mm}
\end{figure}

Fig.~\ref{fig:VHF_validation} separates the available VHF measurement evidence from the rough-sea model predictions. Fig.~\ref{Fig:VHF_a} compares the proposed formulation with maritime path-loss measurements reported at 160~MHz \cite{hao2026measured}, close to the 162~MHz carrier used in the subsequent network analysis. The measurements were collected under calm-sea conditions. At this frequency and at small grazing angles, the sea surface is electrically smooth, so that $\rho_{\text{s}}\approx-1$ and the proposed formulation approaches Lee's smooth-sea model. Both models follow the measured path-loss trend with an RMSE of approximately 2.59~dB, whereas free-space path loss gives an RMSE of 23.41~dB. This agreement is consistent with the calm-sea VHF limiting behavior; it does not constitute validation of roughness-induced VHF path loss for $H_{\text{s}}>1$~m. The 5.15~GHz measurements cannot substitute for such validation because electrical roughness depends explicitly on wavelength. The 160~MHz data are also not used to validate an evaporation-duct-induced third path, whose strength depends on the atmospheric refractivity profile of the measurement campaign. Fig.~\ref{Fig:VHF_b} separately presents the wavelength-specific model prediction at $f_c=162$~MHz, quantified by
	\begin{equation}\label{eq:delta_pl_vhf}
		\Delta PL(d,H_{\mathrm{s}})
		=
		PL_{\mathrm{rough}}(d,H_{\mathrm{s}})
		-
		PL_{\mathrm{smooth}}(d).
	\end{equation}
	For small $H_{\mathrm{s}}$, the predicted deviation from the smooth-sea reference is weak. As $H_{\mathrm{s}}$ increases, attenuation of the coherent specular reflection component reduces the depth of destructive-interference nulls near the multipath fading regions. These curves are wavelength-specific model predictions rather than rough-sea VHF measurement results.

\begin{figure}[!t]
	\centering
	\includegraphics[width=0.5\textwidth]{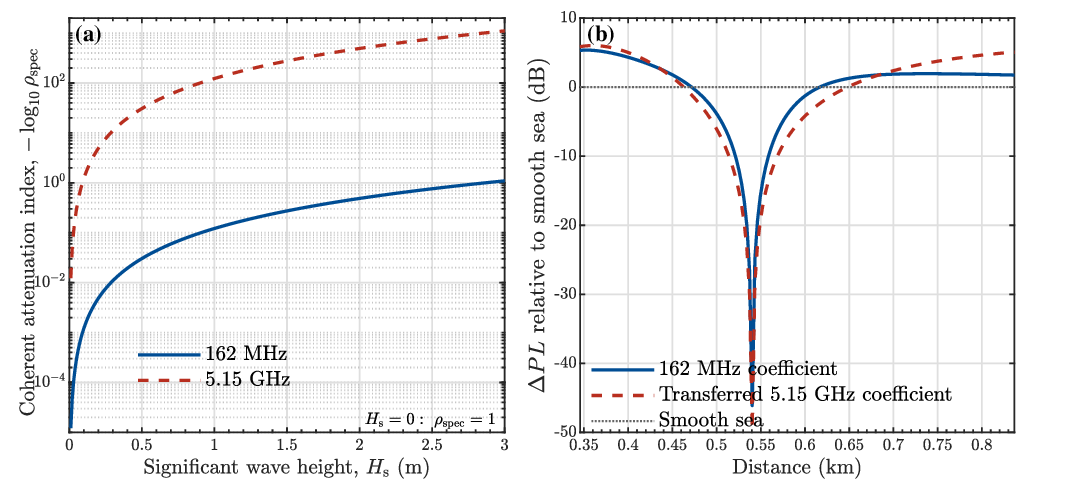}
	\vspace*{-8mm}
	\caption{Cross-frequency coefficient-transfer ablation. Panel (a) compares the logarithmic attenuation index $-\log_{10}\rho_{\mathrm{spec}}$ at 162~MHz and 5.15~GHz as a function of the significant wave height under the same reference geometry. Panel (b) shows the predicted 162~MHz path-loss change $\Delta PL=PL_{\mathrm{rough}}-PL_{\mathrm{smooth}}$, comparing	the wavelength-specific 162~MHz coefficient with a diagnostic transfer of the 5.15~GHz coefficient while retaining all other VHF propagation parameters.}
	\label{fig:frequency_transfer_ablation}
	\vspace*{-5mm}
\end{figure}
\textit{Cross-Frequency Transfer Ablation:} Fig.~\ref{fig:frequency_transfer_ablation} examines the consequence of transferring the C-band coherent-reflection coefficient to VHF. Panel~(a) compares the logarithmic attenuation index $-\log_{10}\rho_{\mathrm{spec}}$ at 162~MHz and 5.15~GHz under the same significant wave height and propagation geometry. The results show that the coherent-reflection attenuation exhibits strong wavelength dependence, and the coefficient obtained at C-band cannot represent the VHF behavior. Panel~(b) retains all 162~MHz propagation parameters, including the path distance, propagation phase, free-space loss, and breakpoint, while replacing only the wavelength-specific VHF coefficient with the coefficient evaluated at 5.15~GHz. The transferred coefficient changes the predicted depth and shape of the VHF fading region near the multipath null. This result demonstrates that the coherent-reflection coefficient must be independently evaluated according to the operating wavelength. This ablation is a model-to-model diagnostic of cross-frequency coefficient transfer and does not replace direct rough-sea VHF measurements.

\begin{figure}[!b]
	\vspace*{-5mm}
	\begin{center}
		\includegraphics[width=1\linewidth]{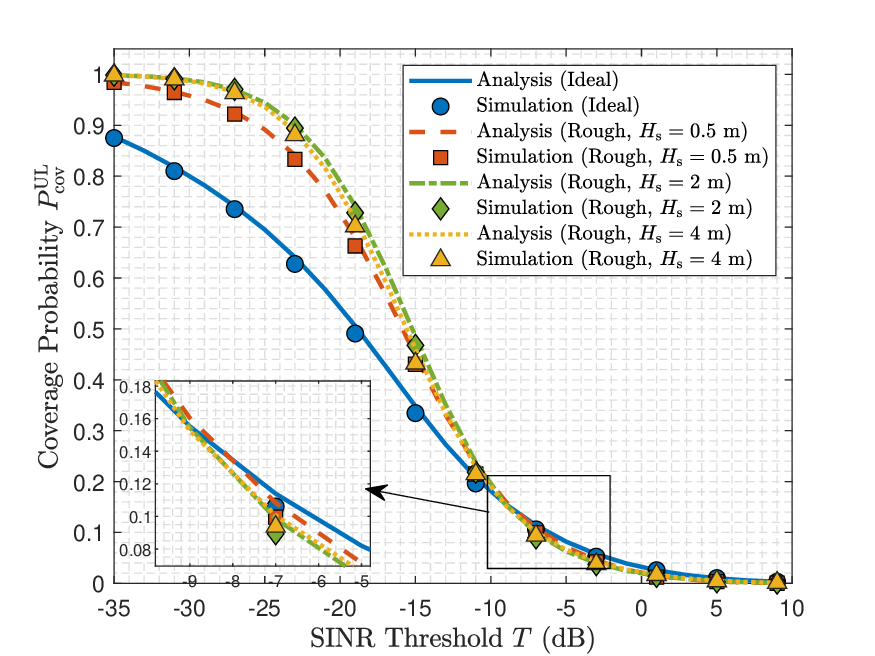}
	\end{center}
	\vspace*{-6mm}
	\caption{Coverage probability as a function of SINR threshold $T$ for different significant wave heights $H_{\textrm{s}}$.}
	\label{fig:6} % Fig.7
	\vspace*{-1mm}
\end{figure}

\subsection{Performance Evaluation and Discussion}\label{S6.2}

{\color{black}
	The following 162~MHz coverage and rate results are analytical and Monte Carlo evaluations under the proposed channel model.
	For rough sea states, particularly $H_{\text{s}}>1$\,m, they represent model-based performance predictions and should not be interpreted as directly measurement-validated performance forecasts for operational VHF maritime systems.
}

The uplink coverage probability $P_{\mathrm{cov}}^{\mathrm{UL}}$, as a function of the SINR threshold $T$, is depicted in Fig.~\ref{fig:6} for the proposed roughness-aware model and Lee's ideal model.  The Monte Carlo simulations agree closely with the analytical results across all test cases, validating the accuracy of the theoretical framework. A key performance crossover is observed in Fig.~\ref{fig:6}. Specifically, in the low-SINR regime, which corresponds to high-reliability requirements, the proposed rough-surface model consistently outperforms the ideal model. Conversely, this trend inverts in the high-SINR regime, where the ideal model exhibits a superior probability of achieving high SINR values. Notably, in the low-SINR region, increasing sea surface roughness proves beneficial.

This behavior is attributed to the dual effect of surface roughness on multipath propagation. The superior performance of the rough-surface model in the low-SINR regime is due to null mitigation. The deep path-loss nulls caused by destructive interference in the ideal model can severely degrade high-reliability links. Our proposed model captures the physical reality that roughness weakens the specular reflection, which substantially reduces the depth of these fading nulls. This raises the floor of the SINR distribution, increasing the probability of exceeding a low threshold. Conversely, the superior performance of the ideal model in the high-SINR regime is due to perfect constructive interference. The same scattering effect that mitigates nulls also reduces the magnitude of constructive signal peaks, which are regions of maximum signal enhancement created by the in-phase summation of the direct and reflected paths. This higher probability of achieving rare, high-SINR events explains why the ideal model's coverage probability is superior when the SINR threshold is high.

\begin{figure}[!b]
	\vspace*{-5mm}
	\begin{center}
		\includegraphics[width=1\linewidth]{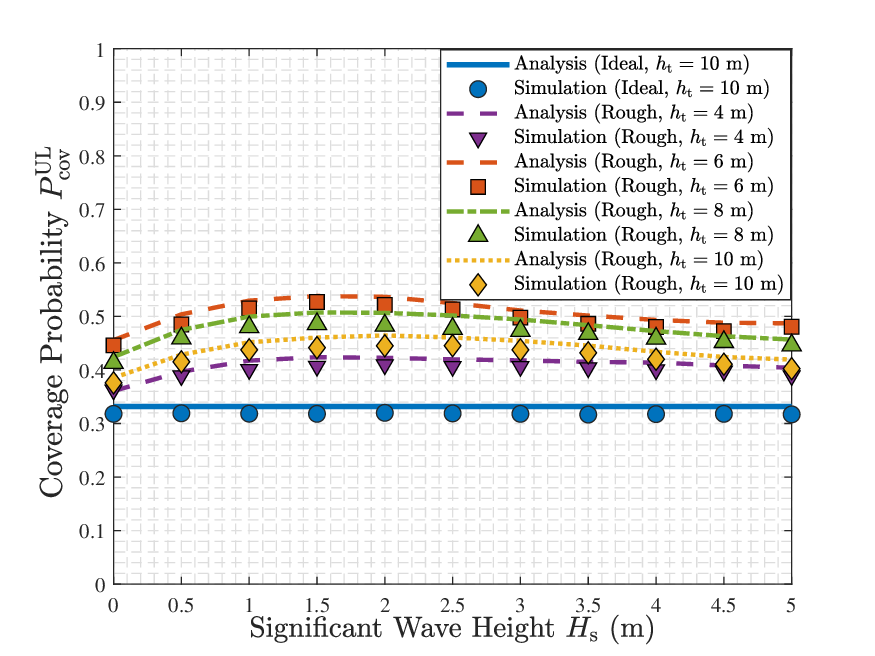}
	\end{center}
	\vspace*{-6mm}
	\caption{Coverage probability as a function of significant wave height $H_{\textrm{s}}$ for different transmitter antenna heights $h_{\textrm{t}}$.}
	\label{fig:7} % Fig.8
	\vspace*{-1mm}
\end{figure}

The uplink coverage probability as a function of significant wave height $H_{\textrm{s}}$ is illustrated in Fig.~\ref{fig:7} for several vessel antenna heights $h_{\textrm{t}}$, evaluated at a fixed SINR threshold of -15\,dB.  Three primary trends are evident from the results. First, our rough-surface model consistently predicts higher coverage than the ideal model, which yields more conservative predictions under the considered conditions. Second, for the considered vessel distribution and propagation geometry, the coverage probability first increases as $h_{\mathrm{t}}$ changes from 4~m to 6~m and then decreases at 8~m and 10~m. Among the tested heights, $h_{\mathrm{t}}=6$~m yields the largest coverage in this specific scenario. Third, under the adopted coherent-reflection attenuation model, the coverage probability varies non-monotonically with the significant wave height and reaches its largest values at approximately 1.5--2.5~m before declining. For $H_{\mathrm{s}}>3$~m, this trend should be interpreted as a model-based sensitivity result because the increasingly important diffuse field is not explicitly resolved.

This non-monotonic behavior arises from the alignment between the multipath interference pattern and the assumed vessel distribution. For the selected parameters, the vessel density peaks at approximately 500~m, and the interference pattern associated with $h_{\mathrm{t}}=6$~m aligns more favorably with this region than those associated with the other tested heights. The preferred height therefore depends on the target service distance and vessel distribution and should not be interpreted as a general antenna-height design rule.

\begin{figure}[!b]
	\vspace*{-5mm}
	\begin{center}
		\includegraphics[width=1\linewidth]{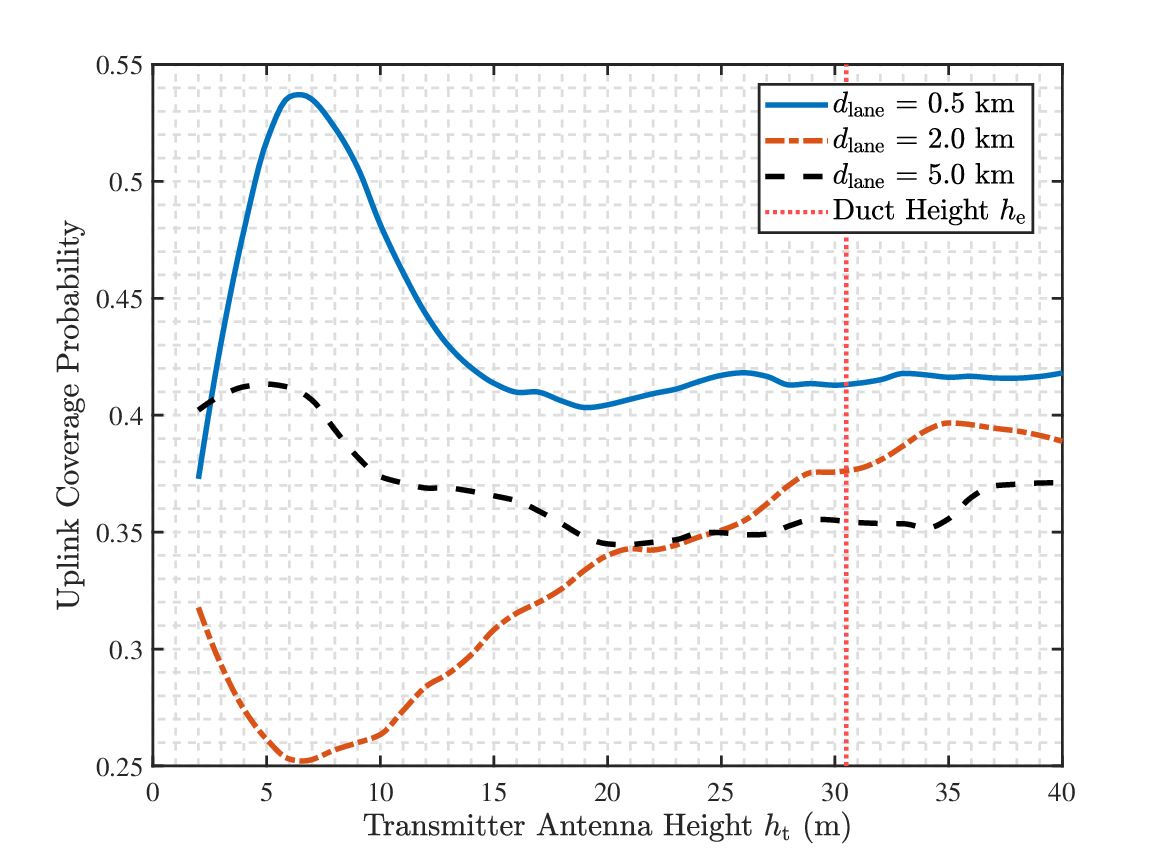}
	\end{center}
	\vspace*{-6mm}
\caption{Uplink coverage probability as a function of the transmitter antenna height $h_{\mathrm{t}}$ for different shipping-lane distances $d_{\mathrm{lane}}$. The height yielding the largest coverage within the tested range varies with the target service distance.}
	\label{fig:robustness}  %fig.9 new
	\vspace*{-1mm}
\end{figure}
To examine the distance dependence of the antenna-height trend, Fig.~\ref{fig:robustness} compares the coverage probability for $d_{\mathrm{lane}}=0.5$, 2.0, and 5.0~km. Among the tested heights, $h_{\mathrm{t}}\approx6$~m yields the largest coverage for $d_{\mathrm{lane}}=0.5$~km, whereas larger heights become more favorable at the longer target distances. This change results from the distance-dependent alignment of the direct and reflected paths with the dominant vessel region. The preferred height is therefore specific to the considered service distance, vessel distribution, and propagation geometry, while the proposed framework provides a means of evaluating this dependence for a given deployment.

\begin{figure}[!b]
	\vspace*{-3mm}
	\begin{center}
		\includegraphics[width=1\linewidth]{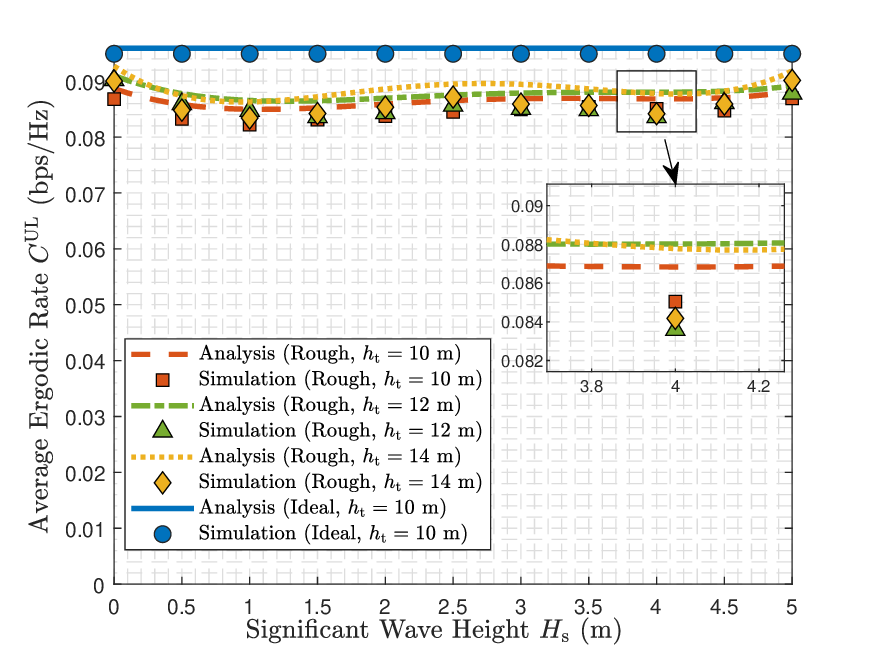}
	\end{center}
	\vspace*{-6mm}
	\caption{Average ergodic rate  as a function of significant wave height $H_{\text{s}}$ for different transmitter antenna heights $h_{\text{t}}$.}
	\label{fig:8} % Fig.9
	\vspace*{-1mm}
\end{figure}

\begin{figure}[!t]
	\vspace*{-2mm}
	\begin{minipage}{0.95\linewidth} %
		\centering
		  \setlength{\subfigcapskip}{-3mm}  
		\subfigure[Coverage probability comparison by regime] {	
			\centering
			\vspace*{-5mm}
			\includegraphics[width=1\linewidth]{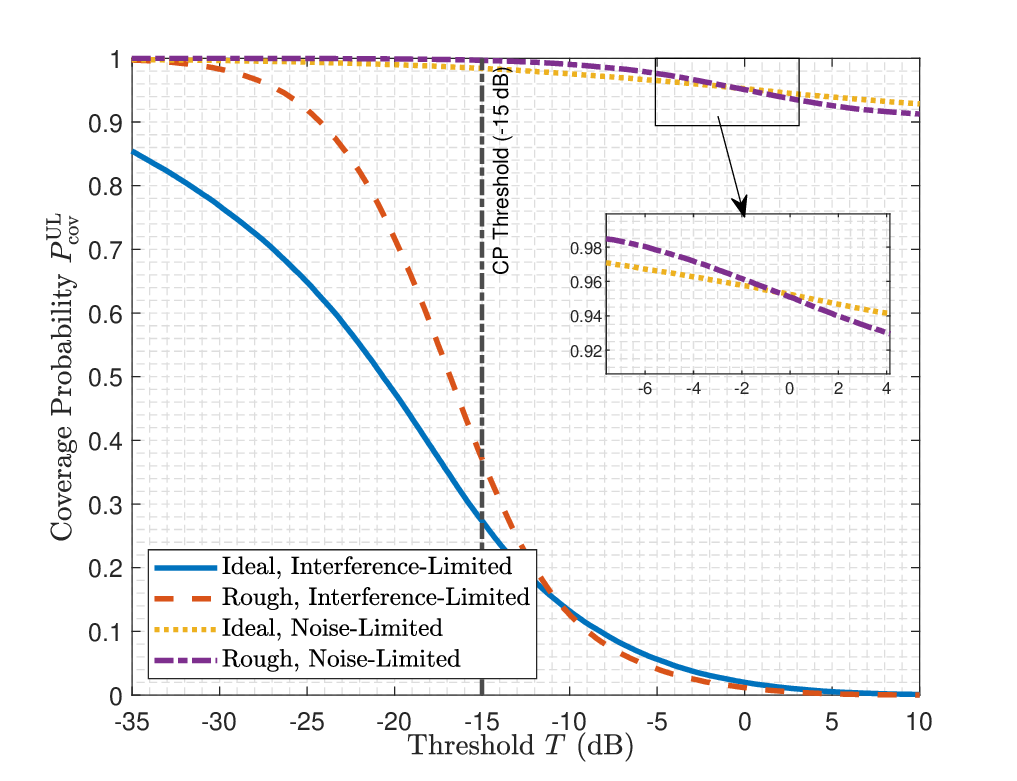}
			\label{Fig. 9a}
		}
			\vspace*{-2mm}
	\end{minipage}  
	\begin{minipage}{0.95\linewidth} %
		\centering
		  \setlength{\subfigcapskip}{-3mm}  
		\subfigure[Coverage probability and average ergodic rate comparison by regime]
		{	
			\centering
			\includegraphics[width=1\linewidth]{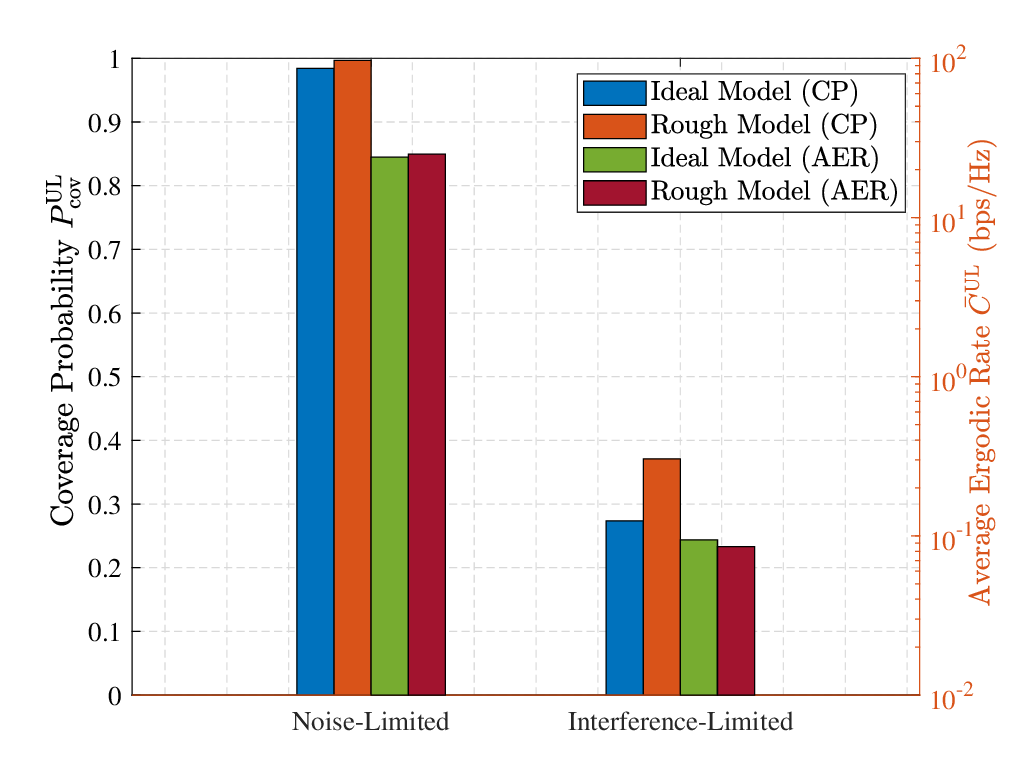}
			\label{Fig. 9b}
		}
		%	\vspace*{-2mm}
	\end{minipage}
	\vspace*{-2mm}
	\caption{Performance comparison in noise- and interference-limited regimes.}
	\label{fig:9} % Fig.10
	\vspace*{-5mm}
\end{figure}

For the tested antenna heights and vessel distribution, Fig.~\ref{fig:8} shows a larger average ergodic rate at higher $h_{\mathrm{t}}$. This ordering results from the corresponding grazing geometry and multipath alignment in the considered scenario and should not be interpreted as a general monotonic antenna-height rule.
Furthermore, Fig.~\ref{fig:8} illustrates the impact of sea surface roughness on the average ergodic rate by benchmarking the proposed model against the ideal smooth-surface model (Lee's model \cite{Lee20214}). As observed, the average ergodic rate of the proposed model remains consistently lower than the ideal baseline. For a fixed $h_{\mathrm{t}}$, the adopted model predicts a non-monotonic rate variation, with an initial decrease followed by a partial recovery as $H_{\mathrm{s}}$ increases. The portion of this behavior occurring above $H_{\mathrm{s}}=3$~m should be interpreted cautiously because the Rayleigh factor suppresses only the coherent specular reflection component, while the dominant diffuse field is represented only through a coarse statistical approximation.
	
Taken together, Figs.~\ref{fig:7} and~\ref{fig:8} reveal a conditional reliability--capacity trade-off under the adopted coherent-reflection attenuation model. Moderate attenuation of the coherent specular reflection component can mitigate destructive-interference nulls and improve low-threshold coverage, while the reduction in constructive reflected power can decrease the average ergodic rate. Under severe sea states, this interpretation remains subject to the limitation that the diffuse field is represented only approximately. As observed, while moderate roughness enhances the coverage probability (Fig.~\ref{fig:7}), it simultaneously reduces the average ergodic rate (Fig.~\ref{fig:8}). This phenomenon is attributed to the effect of channel hardening. Specifically, surface roughness attenuates the coherent specular reflection, thereby reducing the variance of the received signal power. This suppression mitigates the deep fading nulls caused by destructive interference, thus raising the minimum SINR floor and improving link reliability. Conversely, it attenuates the strong constructive interference peaks that contribute disproportionately to the ergodic capacity, leading to a reduction in the average ergodic rate. In the context of safety-critical maritime communications, where continuous connectivity is prioritized over peak throughput, this trade-off indicates that moderate roughness may improve link reliability under the considered propagation and traffic conditions, although at the expense of average rate.

%\begin{figure}[!t]
%	\vspace*{-1mm}
%	\begin{center}
%		\includegraphics[width=1\linewidth]{verify_AERandPcov.eps}
%	\end{center}
%	\vspace*{-5mm}
%	\caption{Comparative analysis of Coverage Probability (CP) and Average Ergodic Rate (AER) for the proposed roughness-aware model and the classical Lee model.}
%	\label{fig:9}
%	\vspace*{-6mm}
%\end{figure}

	Fig.~\ref{fig:9} shows that the relative performance of the two models depends strongly on the operational regime, reflecting a reliability-capacity trade-off. In the interference-limited regime, the proposed roughness-aware model improves coverage reliability, whereas Lee's ideal model yields a higher average ergodic rate. This behavior is consistent with the standard channel-hardening effect. The ideal model preserves strong constructive interference, resulting in a heavy-tailed SINR distribution in which infrequent high-SINR peaks disproportionately increase the logarithmic average rate. Conversely, the proposed model accounts for the suppression of the coherent specular reflection component by surface roughness, which reduces the variance of the received signal power. By mitigating both deep destructive nulls and high constructive peaks, the SINR distribution becomes more compressed. This hardening effect improves the probability of exceeding moderate SINR thresholds, albeit at the cost of reducing rare high-SINR events that contribute to peak capacity.

%A comparative analysis in Fig.~\ref{fig:9} reveals that the relative performance in terms of average ergodic rate is critically dependent on the operational regime. In the interference-limited regime, a performance paradox emerges: while the proposed model offers substantially higher reliability, the classical Lee's model yields a greater average ergodic rate. This is a direct consequence of the ideal model's heavy-tailed SINR distribution, an artifact of perfect constructive interference, where infrequent but high-SINR events disproportionately elevate the logarithmic average. Conversely, this average ergodic rate trend inverts under noise-limited conditions. Here, the proposed model results in a higher average ergodic rate because the severe signal nulls inherent to the Lee model become the dominant detrimental factor, and the proposed model's ability to mitigate these deep fades leads to its superior performance.

\section{Conclusion}\label{S7}

This paper has developed a tractable stochastic-geometry framework for near-shore maritime communications by integrating a roughness-dependent coherent-reflection coefficient with a non-homogeneous vessel distribution and aggregate uplink interference. Analytical expressions for the uplink coverage probability and average ergodic rate have been derived under the resulting piecewise maritime propagation model.

The results show that attenuation of the coherent specular reflection component can mitigate deep fading nulls and produce a conditional reliability--capacity trade-off. The 5.15~GHz measurements assess the roughness-sensitive reflection mechanism, whereas the 160~MHz measurements examine only the calm-sea VHF limit; the 162~MHz rough-sea results therefore remain model-based predictions. The cross-frequency ablation further confirms that a coherent-reflection coefficient evaluated at C-band cannot be transferred unchanged to VHF. Under severe sea states, particularly for $H_{\mathrm{s}}>3$~m, characterizing the unresolved diffuse field via Nakagami-$m$ fading remains a coarse statistical approximation. Future work will consider synchronized VHF measurements, explicit diffuse-scattering models, protocol-aware dependent access, vessel mobility, Doppler effects, and multi-station interference.
\appendix

\vspace*{-2mm}
\subsection{Proof of Theorem~\ref{T1}}\label{ApA}

\begin{proof}
\emph{1)}~The two-ray path loss model is derived from the coherent superposition of the electric fields from the direct LOS path and the sea-reflected path.
	
The total received electric field, $E_{\text{total}}$, is the vector sum of the LOS field, $E_{\text{LOS}}$, and the reflected field, $E_{\text{refl}}$:
\begin{equation}\label{eqApA1} % eq.27
	E_{\text{total}} = E_{\text{LOS}} + E_{\textrm{refl}}.
\end{equation}
Assuming an initial transmitted field of $E_0$, the fields arriving at the receiver can be expressed as:
\begin{align} % eqs.28,29
	E_{\text{LOS}} =& \frac{E_0}{d_1} \exp ({-jkd_1}) , \label{eqApA2} \\
	E_{\text{refl}} =& \rho_{\text{s}} \cdot \frac{E_0}{d_2} \exp ({-jkd_2}), \label{eqApA3}
\end{align}
where $d_1$ and $d_2$ are the path lengths of the direct and reflected rays, respectively, $k=2\pi/\lambda$ is the wavenumber, and $\rho_{\text{s}}$ is the effective reflection coefficient of the sea surface.
	
For far-field scenarios, where the horizontal distance $d$ is significantly greater than the antenna heights ($d \gg h_{\text{t}}, h_{\text{r}}$), we apply two standard approximations. First, for the amplitude attenuation, we approximate $d_1 \approx d_2 \approx d$. Second, for the phase calculation, the path length difference, $\Delta d$, is approximated using a Taylor expansion as:
\begin{equation} % eq.30
	\Delta d = d_2 - d_1 \approx \frac{2 h_{\text{t}} h_{\text{r}}}{d}. \label{eq:path_diff_2ray}
\end{equation}
The total field can thus be written as:
\begin{equation}\label{eqApA5} % eq.31
	E_{\text{total}} \approx \frac{E_0}{d} \left( \exp ({-jkd_1}) + \rho_{\text{s}} \exp ({-jkd_2}) \right).
\end{equation}
	
The received power, $P_{\textrm{r}}$, is proportional to the squared magnitude of the total electric field. By factoring out the common phase term $\exp ({-jkd_1})$, we get:
\begin{align}\label{eqApA6} % eq.32
	|E_{\text{total}}|^2 \approx& \left| \frac{E_0}{d} \exp ({-jkd_1}) \left( 1 + \rho_{\text{s}} \exp ({-jk\Delta d}) \right) \right|^2 \nonumber \\
	\overset{\mathrm{(a)}}{=}& \frac{|E_0|^2}{d^2} \left| 1 + \rho_{\text{s}} \exp ({-jk\Delta d}) \right|^2,
\end{align}
where (a) is due to $|\exp ({-jkd_1})|=1$.
Since $\rho_{\text{s}}$ is a negative real number, the squared magnitude term expands to:

\begin{align}\label{eqApA7} % eq.33
	& \left| 1 + \rho_{\text{s}} \exp ({-jk\Delta d}) \right|^2 \nonumber \\
	& \hspace*{5mm}= \left| 1 - |\rho_{\text{s}}|(\cos(k\Delta d) - j\sin(k\Delta d)) \right|^2 \nonumber \\
	& \hspace*{5mm}= 1 + \rho_{\text{s}}^2 - 2|\rho_{\text{s}}|\cos(k\Delta d).
\end{align}
	
The path loss is the ratio of transmitted power to received power. For isotropic antennas in free space, $P\!L_{\text{fs}}\! =\! (4\pi d / \lambda)^2$. The total path loss is the free-space path loss divided by the multipath power gain factor, $G_{\text{mp}}^{\textrm{2-ray}}\! =\! |1 + \rho_{\text{s}} \exp ({-jk\Delta d})|^2$. Using \eqref{eq:path_diff_2ray} and $k=2\pi/\lambda$ yields the final expression of $P\!L_{\textrm{2-ray}}^{\textrm{rough}}(d)$ given in (\ref{eq3-4}). 

\emph{2)}~The three-ray model extends the two-ray model by including a third path component, $E_{\text{duct}}$, which results from refraction by the evaporation duct. This path is approximated as a reflection from an effective duct height $h_{\text{e}}$.
The total received electric field is the sum of three components:
\begin{equation}\label{eqApA8} % eq.34
	E_{\text{total}} = E_{\text{LOS}} + E_{\text{refl}} + E_{\text{duct}}.
\end{equation}
Using the far-field amplitude approximation, this becomes:
\begin{align}\label{eqApA9} % eq.35
	E_{\text{total}} \approx& \frac{E_0}{d} \big( \exp ({-jkd_1}) + \rho_{\text{s}} \exp ({-jkd_2}) \nonumber \\
	& + \Gamma_{\text{duct}} \exp ({-jkd_3}) \big),
\end{align}
where $d_3$ is the path length of the duct-refracted ray, and $\Gamma_{\text{duct}}$ is its reflection coefficient, assumed to be ideal ($\Gamma_{\text{duct}} = -1$) as per the model in \cite{Lee20214}.
	
The path length difference for the sea-reflected ray, $\Delta d_2\! =\! d_2 - d_1$, is the same as in \eqref{eq:path_diff_2ray}. The path length difference for the duct-refracted ray, $\Delta d_3$, is derived from the geometry of a reflection from height $h_{\text{e}}$:
\begin{equation}\label{eqApA10} % eq.36
	\Delta d_3 = d_3 - d_1 \approx \frac{2(h_{\text{e}} - h_{\text{t}})(h_{\text{e}} - h_{\text{r}})}{d}.
\end{equation}
The multipath power gain factor, $G_{\textrm{mp}}^{\textrm{3-ray}}$, is the squared magnitude of the term in the parentheses:
\begin{align}\label{eqApA11} % eq.37
	G_{\textrm{mp}}^{\textrm{3-ray}} &= \left| 1 + \rho_{\text{s}} \exp ({-jk\Delta d_2}) - \exp ({-jk\Delta d_3}) \right|^2 \nonumber \\
	&= (1 + \rho_{\text{s}} C_2 - C_3)^2 + (\rho_{\text{s}} S_2 - S_3)^2
\end{align}
where $C_2=\cos(k\Delta d_2)$, $S_2=\sin(k\Delta d_2)$, $C_3=\cos(k\Delta d_3)$,  and $S_3=\sin(k\Delta d_3)$. The full expansion provides the complete gain factor. The final path loss is given by:
\begin{equation}\label{eqApA12} % eq.38
	P\!L_{\textrm{3-ray}}^{\textrm{rough}}(d) = \frac{P\!L_{\text{fs}}}{G_{\textrm{mp}}^{\textrm{3-ray}}}.
\end{equation}
Noting $P\!L_{\text{fs}}\! =\! (4\pi d / \lambda)^2$ completes the proof.
\end{proof}

\vspace*{-2mm}
\subsection{Proof of Lemma~\ref{L1}}\label{ApB}

\begin{proof}
\begin{align} \label{eqApB1} % eq.39
	\mathcal{L}_{I, \textrm{near}} (s) 	& = \left. \mathbb{E}\! \left[ \exp\!\left( -s \sum_{S_i \in\Phi_{\textrm{S}}' } \frac{P_{S_i} \left|h_{i}\right|^2}{P\!L(D_i)}\right) \! \right]\right|_{D_i=d_i} \nonumber \\
	& \hspace*{-5mm}= \mathbb{E}_{\Phi_{\textrm{S}}'}\!\Bigg[\! \prod_{S_i \in\Phi_{\textrm{S}}' }\!\! \mathbb{E}_{|h_{i}|^2}\bigg[\exp\bigg(\left|h_{i}\right|^2 \bigg(-s \frac{P_{S_i} }{P\!L_{\textrm{2-ray}}(d_i)}\bigg)\! \bigg]\Bigg] \nonumber\\
	& \hspace*{-5mm}\overset{\mathrm{(a)}}{=} \mathbb{E}_{\Phi_{\textrm{S}}'}\!\Bigg[\! \prod_{S_i \in\Phi_{\textrm{S}}' } {\left(\!1\!+\!\frac{s P_{S_i}}{m \, P\!L_{\textrm{2-ray}}(d_i) }\!\right)^{\!\!\!-m}} \!\Bigg]\nonumber\\
	& \hspace*{-5mm}\overset{\mathrm{(b)}}{=} \exp \left( \int_S \left[{\left(\!1\!+\!\frac{s P_{S_i}}{m \, P\!L_{\textrm{2-ray}}(d_i) }\!\right)^{\!\!\!-m}}-1\right]\mu'(d_i) \,\mathrm{d}\mathbf{x} \right) \nonumber \\
	& \hspace*{-5mm}\overset{\mathrm{(c)}}{=} \exp \left(\pi \int_0^{d_{\textrm{break}}} \left[ {\left(\!1\!+\!\frac{s P_{S_i}}{m \, P\!L_{\textrm{2-ray}}(d_i) }\!\right)^{\!\!\!-m}}-1 \right] \right. \nonumber \\
	& \times p_a \mu_0 d_i^{\alpha} \exp ({-\beta d_i}) \,\mathrm{d}d_i \Bigg) ,
\end{align}  
where (a) is obtained by using the moment-generating function (MGF) of the normalized Gamma random variable, (b) is given by the probability generating functional (PGFL) of the PPP, and (c) is obtained by replacing $\mathrm{d}\mathbf{x}$ with $\pi d_i \,\mathrm{d}d_i$.
	
Similarly, we can obtain the Laplace transform of the interference from vessels in the range $d_{\textrm{break}}$ to $d_{\textrm{LOS}}$:
\begin{align} \label{eqApB2} % eq.40
	\mathcal{L}_{I, \mathrm{far}}(s_1)	&= \exp \left( \pi \int_{d_{\textrm{break}}}^{d_{\textrm{LOS}}} \left[ {\left(\!1\!+\!\frac{s P_{S_i}}{m \, P\!L_{\textrm{3-ray}}(d_i) }\!\right)^{\!\!\!-m}}-1 \right] \right.\nonumber
\end{align}
\begin{align}
	&\,\,\,\,\times p_a \mu_0  d_i^{\alpha} \exp ({-\beta d_i}) \,\mathrm{d}d_i \Bigg)
\end{align}
This completes the proof.
\end{proof}

\subsection{Proof of Theorem~\ref{The2}}\label{ApC}

\begin{proof}
\begin{align}\label{eqApC1} % eq.41
	P_{\mathrm{cov}}^{\textrm{UL}} \triangleq& \,\, \mathbb{P}( \mathrm{SINR}_{S_m}\geq T) \nonumber \\
	=& \, \mathbb{E}\left[ \mathbb{P}\left(\mathrm{SINR}_{S_m} \geq T|D_{S_m}=d\right) \right] \nonumber 
\end{align}
\begin{align}
	= & \int_0^{d_{\mathrm{break}}} \mathbb{P}\left(\mathrm{SINR}_{S_m}^{\textrm{2-ray}} \geq T|d\right)  g(d) \,\mathrm{d}d \nonumber \\
	+ & \int_{d_{\mathrm{break}}}^{d_{\mathrm{LOS}}} \mathbb{P}\left(\mathrm{SINR}_{S_m}^{\textrm{3-ray}} \geq T|d\right)  g(d) \,\mathrm{d}d .
\end{align}

First we note:
\begin{align*}
	& \mathbb{P}\left(\mathrm{SINR}_{S_m}^{\textrm{2-ray}} \geq T|d\right) \nonumber \\
	& \hspace*{5mm} = 1 - \mathbb{P}\left(|h_{S_m}|^2 \le \frac{T\cdot (I_{S_m}\!\!+N_0) \cdot P\!L_{\textrm{2-ray}}(d)}{	P_{S_m}} \right)
\end{align*}
\begin{align}\label{eqApC2} % eq.42
	& \overset{\mathrm{(a)}}{\approx}  1 - \mathbb{E}\left[\left(1 - \exp\left(\frac{-\eta T\cdot (I_{S_m}\!\!+N_0) \cdot P\!L_{\textrm{2-ray}}(d)}{	P_{S_m}} \right) \right)^{m}\right]  \nonumber \\
	& \overset{\mathrm{(b)}}{=} \!\mathbb{E}_I \!\!\left[\! \sum_{n=1}^{m}(-1)^{n+1}\!\binom{m}{n} \!\exp \! \left(\frac{-n \eta T (I_{S_m}\!\!+N_0)  P\!L_{\textrm{2-ray}}(d)}{	P_{S_m}} \right)\!\!\right] \nonumber 
\end{align}
\begin{align}
	& \overset{\mathrm{(c)}}{=} \sum_{n=1}^{m}(-1)^{n+1}\binom{m}{n} \mathbb{E}_I\left[ \exp\left(-s{I_{S_m}}\right)\right]\,\, \exp\left(-s{N_0} \right)\nonumber\\
	& = \sum_{n=1}^{m}(-1)^{n+1}\binom{m}{n}\mathcal{L}_{I, \mathrm{near}}\! \left(s\right) \,\,\exp\left(-s{N_0} \right),
\end{align}
where (a) is a tight upper bound when $m$ is small \cite{alzer1997some}, that is, for small $m$, $\mathbb{P}\left(|h|^{2}<\psi\right)<\mathbb{E}\left[\left(1 - \exp (-\psi\eta)\right)^{m}\right]$ with $\eta=m(m!)^{-\frac{1}{m}}$, and (b) is obtained by the binomial theorem, while (c) is obtained by denoting $s= \frac{-n \,\eta \,T\,  P\!L_{\textrm{2-ray}}(d)   }  {P_{S_m}} $.
	
Similarly, $\mathbb{P}\left(\mathrm{SINR}_{S_m}^{\textrm{3-ray}} \geq T|d\right)$ can be expressed as:
\begin{align}\label{eqApC3} % eq.43
	& \mathbb{P}\left(\mathrm{SINR}_{S_m}^{\textrm{3-ray}} \geq T|d\right) \nonumber\\
	& \hspace*{5mm}=\sum_{n=1}^{m}(-1)^{n+1} \binom{m}{n} \mathcal{L}_{I, \mathrm{far}} \left(s_1 \right) \exp  \left(-s_1 {N_0} \right),
\end{align}
where $s_1= \frac{-n \,\eta \,T\,P\!L_{\textrm{3-ray}}(d)}{P_{S_m}}$.

We obtain $P_{\textrm{cov}}^{\textrm{UL}}$ of (\ref{eqThe21}) by substituting (\ref{eqgd}) (\ref{eqL1near}), (\ref{eqL1far}), (\ref{eqApC2}) and (\ref{eqApC3}) into (\ref{eqApC1}). This completes the proof.
\end{proof}

%--------------------------------------------
%\iffalse
%\begin{align}\label{eqThe21}
%	&P_{\mathrm{cov}}^{\mathrm{UL}}  \nonumber \\
%	&= \int_0^{d_{\rm{break}}} \,\sum_{n=1}^{m}(-1)^{n+1}\binom{m}{n} \nonumber\\
%	&\,\,\,\,\times\exp \left(2\pi \int_0^{d_{\textrm{break}}} \left[ {\left( \!1\!+\!\frac{s P_{S_i} 4\lambda^2 \sin^2\left(\frac{2\pi h_{\textrm{t}} h_{\textrm{r}}}{\lambda d_i}\right)}{(4\pi d_i)^2 m }\!\right)^{\!\!\!-m}}\!\!\! -1 \right] \right.\nonumber\\
%	&\,\,\,\,\times p_a \mu_0 f(d_i) d_i \,\mathrm{d}d_i \Bigg)   \,\,\exp\left(-s{N_0} \right)
%	g(d) \,\mathrm{d}d \nonumber \\
%	&\,\,\,\,+ \int_{d_{\mathrm{break}}}^{d_{\rm{LOS}}}  	\sum_{n=1}^{m}(-1)^{n+1} \binom{m}{n} \nonumber\\
%	&\,\,\,\,\times\exp \left(2\pi \int_{d_{\textrm{break}}}^{d_{\textrm{LOS}}} \left[ {\left(\!1\!+\!\frac{s_1 P_{S_i} 4\lambda^2 \left(1+\Delta(d_i)\right)^2}{ (4\pi d_i)^2  m  }\!\right)^{\!\!\!-m}}\!\!-1 \right] \right.\nonumber\\
%	&\,\,\,\,\times p_a \mu_0 f(d_i) d_i \,\mathrm{d}d_i \Bigg)  \exp  \left( -s_1 {N_0}  \right)  g(d) \,\mathrm{d}d,
%\end{align}%%
%\fi	

\subsection{Proof of Theorem~\ref{T3}}\label{ApD}

\begin{proof}
From (\ref{eqAER}), we have
\begin{align}%\label{eqApD1} % eq.44
	\bar{C}^{\textrm{UL}} 	&\triangleq \frac{1}{K}\mathbb{E}\left[\log_2(1 + \mathrm{SINR}_{S_m})\right] \nonumber \\
	& \hspace*{-2mm}\overset{\mathrm{(a)}}=\! \frac{1}{K}\int \!\!\!\! \int_{t>0}\!\!\!\! \mathbb{E}\left[\mathbb{P}\!\!\ \Big(\log_2 (1+ \mathrm{SINR}_{S_m} >t|d\Big) \right]\!  g(d)\,\mathrm{d}t\,\mathrm{d}d \nonumber\\
	& \hspace*{-2mm}= \frac{1}{K}\! \int_{t>0}\! \left(\int_0^{d_{\rm{break}}}\!\! \mathbb{P}\!\!\ \Big(\log_2 (1\! +\! \mathrm{SINR}_{S_m}^{\textrm{2-ray}}) >t|d\Big)   g(d) \,\mathrm{d}d \right.\nonumber 
\end{align}
\begin{align}\label{eqApD1} % eq.44
		& \hspace*{-2mm}+ \left. \int_{d_{\mathrm{break}}}^{d_{\rm{LOS}}}  \mathbb{P}\!\!\ \Big(\log_2 (1+ \mathrm{SINR}_{S_m}^{\textrm{3-ray}}) >t|d\Big) g(d) \,\mathrm{d}d \right)\mathrm{d}t \nonumber\\
	& \hspace*{-2mm}= \frac{1}{K}\int_{t>0}\left(\int_0^{d_{\rm{break}}} \mathbb{P}\!\!\ \Big( \mathrm{SINR}_{S_m}^{\textrm{2-ray}} >2^t-1|d\Big)   g(d) \,\mathrm{d}d \right.\nonumber \\
	& \hspace*{-2mm}+ \left. \int_{d_{\mathrm{break}}}^{d_{\rm{LOS}}}  \mathbb{P}\!\!\ \Big( \mathrm{SINR}_{S_m}^{\textrm{3-ray}} >2^t-1|d\Big) g(d) \,\mathrm{d}d \right)\mathrm{d}t ,
\end{align}	
where (a) follows from the identity $\mathbb{E}[X]=\int_{t>0}\mathbb{P}(X>t)\mathrm{d}t$ for a nonnegative random variable $X$.

First we investigate $\mathbb{P}\left(\mathrm{SINR}_{S_m}^{\textrm{2-ray}} \geq 2^t-1|d\right)$:
\begin{align*}
	& \mathbb{P}\left(\mathrm{SINR}_{S_m}^{\textrm{2-ray}} \geq 2^t-1|d\right) \nonumber \\
	& \! = 1 - \mathbb{P}\left(|h_{S_m}|^2 \le \frac{(2^t-1)\cdot (I_{S_m}\!\!+N_0) \cdot P\!L_{\textrm{2-ray}}(d)}{	P_{S_m}} \right)\nonumber \\
	&\! \overset{\mathrm{(a)}}{\approx}  1 - \mathbb{E}\left[\left(1 - \exp\!\left(\frac{-\eta (2^t-1) (I_{S_m}\!\!+N_0) P\!L_{\textrm{2-ray}}(d)}  {P_{S_m} } \right)\!\! \right)^{\!m}\right]  \nonumber \\
	&\! \overset{\mathrm{(b)}}{=} \!\!\mathbb{E}_{\!I} \!\!\left[\! \sum_{n=1}^{m}(-1)^{n\!+\!1}\!\binom{m}{n}\!\! \exp\!\! \left(\! \!\frac{-n \eta (2^t\! -\! 1)\! (I_{S_m}\!\!+\! N_0)\!P\!L_{\textrm{2-ray}}\!(d)}  {P_{S_m}\, }\!\!\right)\!\!\right]\\
		&\! \overset{\mathrm{(c)}}{=} \sum_{n=1}^{m}(-1)^{n+1}\binom{m}{n} \mathbb{E}_I\left[ \exp\left(-s'{I_{S_m}}\right)\right]\,\, \exp\left(-s'{N_0} \right)\nonumber\\
		&	\! = \sum_{n=1}^{m}(-1)^{n+1}\binom{m}{n}\mathcal{L}_{I, \mathrm{near}}\! \left(s'\right) \,\,\exp\left(-s'{N_0} \right)\nonumber
\end{align*}
\begin{align}\label{eqApD2} % eq.4
	&	\!=  \,\sum_{n=1}^{m}(-1)^{n+1}\binom{m}{n} \exp\left(-s'{N_0} \right)\nonumber\\
	&\,\,\,\,\times\exp \left(2\pi \int_0^{d_{\textrm{break}}} \left[ {\left( \!1\!+\!\frac{s' P_{S_i} }{ m P\!L_{\textrm{2-ray}}\!(d_i) }\!\right)^{\!\!\!-m}}\!\!\! -1 \right] \right.\nonumber\\
	&\,\,\,\,\times p_a \mu_0 d_i^{\alpha} \exp ({-\beta d_i}) \,\mathrm{d}d_i \Bigg)   \,\,,
\end{align}
where (a) is a tight upper bound when $m$ is small \cite{alzer1997some}, that is, for small $m$, $\mathbb{P}\left(|h|^{2}<\psi\right)<\mathbb{E}\left[\left(1 - \exp (-\psi\eta)\right)^{m}\right]$ with $\eta=m(m!)^{-\frac{1}{m}}$, and (b) is obtained by the binomial theorem, while (c) is obtained by denoting $s'= \frac{-n \,\eta \,(2^t-1)\, P\!L_{\textrm{2-ray}}\!(d)    }  {P_{S_m}\ } $.

Similarly, $\mathbb{P}\left(\mathrm{SINR}_{S_m}^{\textrm{3-ray}} \geq 2^t-1|d\right)$ can be expressed as:
\begin{align}\label{eqApD3} % eq.46
	&\mathbb{P}\left(\mathrm{SINR}_{S_m}^{\textrm{3-ray}} \geq 2^t-1|d\right) \nonumber\\
	&=
	\sum_{n=1}^{m}(-1)^{n+1} \binom{m}{n} \mathcal{L}_{I, \mathrm{far}} \left(s_1' \right) \exp  \left( -s_1' {N_0}  \right)\nonumber\\
	&= 	\sum_{n=1}^{m}(-1)^{n+1} \binom{m}{n} \exp  \left( -s_1' {N_0}  \right)\nonumber\\
	&\,\,\,\,\times\exp \left(2\pi \int_{d_{\textrm{break}}}^{d_{\textrm{LOS}}} \left[ {\left(\!1\!+\!\frac{s_1' P_{S_i} 4\lambda^2 \left(1+\Delta(d_i)\right)^2}{ (4\pi d_i)^2  m  }\!\right)^{\!\!\!-m}}\!\!-1 \right] \right.\nonumber\\
	&\,\,\,\,\times p_a \mu_0 d_i^{\alpha} \exp ({-\beta d_i}) \,\mathrm{d}d_i \Bigg),
\end{align}
where $s_1'= {-n \,\eta \,(2^t-1)\,P\!L_{\textrm{3-ray}}(d_i)} / {P_{S_m}}$. 

We obtain $\bar{C}^{\textrm{UL}}$ of (\ref{eqT4}) by substituting (\ref{eqgd}) (\ref{eqL1near}), (\ref{eqL1far}), (\ref{eqApD2}) and (\ref{eqApD3}) into (\ref{eqApD1}). This completes the proof.
\end{proof}	

%\bibliography{IEEEtran}
%\bibliography{reference}	
\small
%\bibliography{IEEEabrv, reference}
%\bibliography{reference}		

\makeatletter
\def\@IEEEBIOskipN{0.5\baselineskip}
\makeatother
\begin{IEEEbiography}
	[{\includegraphics[width=1in,height=1.25in,clip,keepaspectratio]{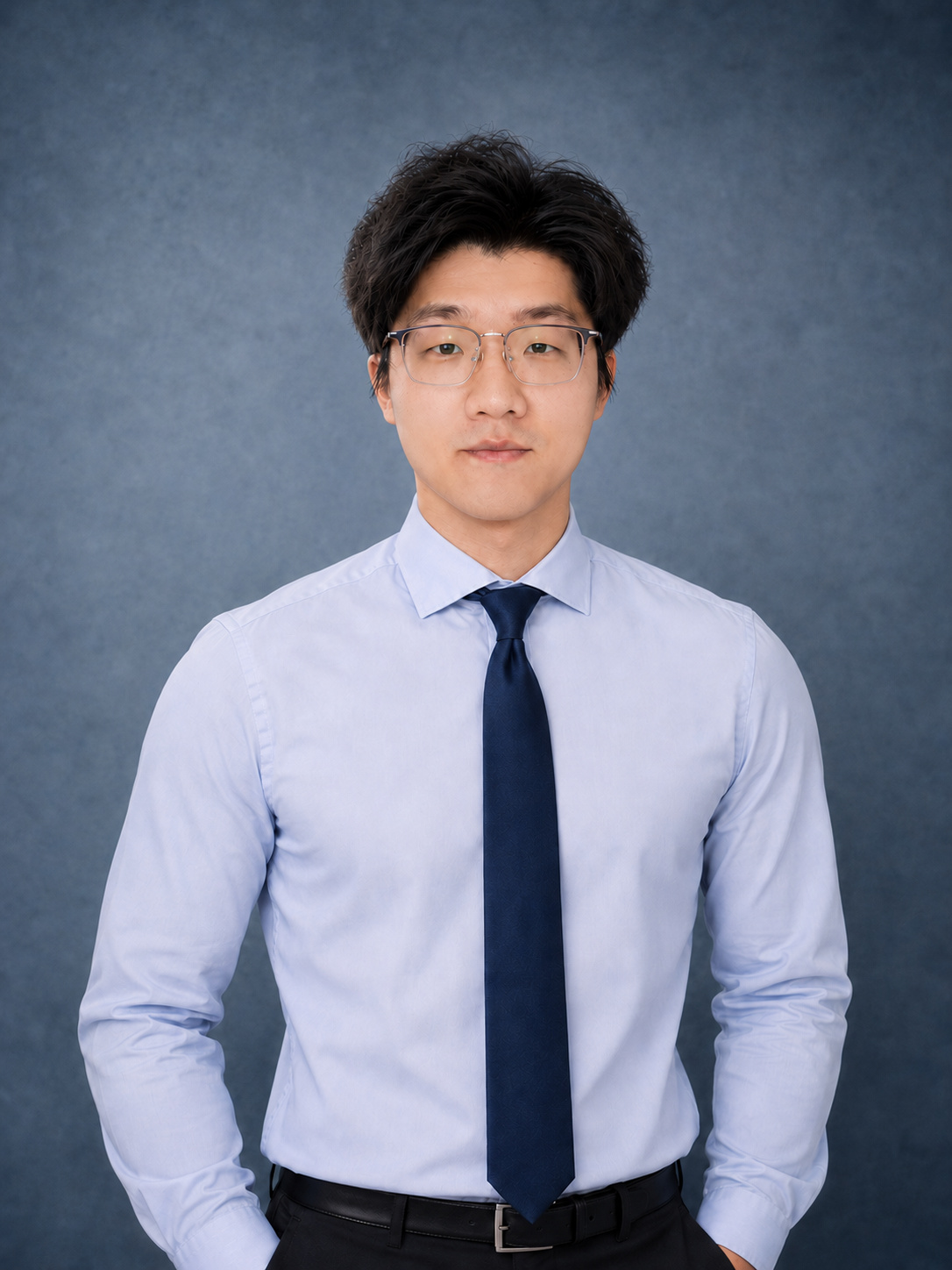}}]
	{Wen-Yu Dong}
	received the B.S. degree in electronic and information engineering from Sichuan University, Chengdu, China, in 2019, and the Ph.D. degree in information and communication engineering from the School of Information and Communication Engineering, Beijing University of Posts and Telecommunications, Beijing, China, in 2025. He is currently a Researcher with the Future Technology Research Center, China Telecom Research Institute, Beijing, China. His current research interests include space-air-ground integrated networks, fluid-spatiotemporal stochastic geometry (F-STSG), and wireless communications.
\end{IEEEbiography}

\begin{IEEEbiography}
	[{\includegraphics[width=1in,height=1.25in,clip,keepaspectratio]{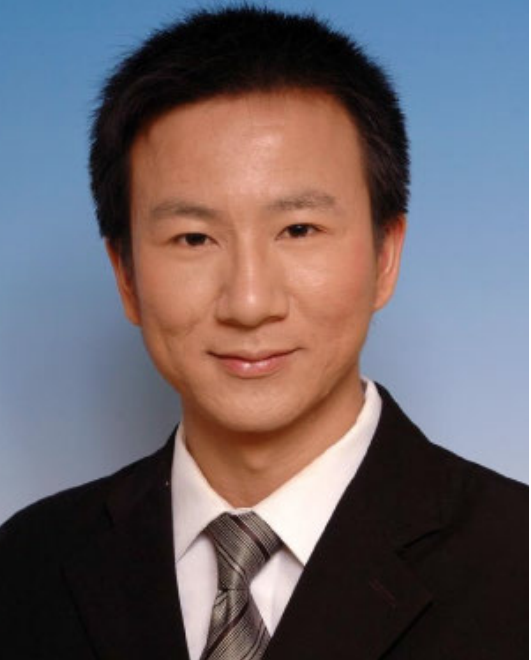}}]
	{Shaoshi Yang}
	(Senior Member, IEEE) received the B.Eng. degree in information engineering from Beijing University of Posts and Telecommunications (BUPT), China, in 2006, and the Ph.D. degree in electronics and electrical engineering from the University of Southampton, U.K., in 2013. From 2008 to 2009, he was a Researcher with Intel Labs China. From 2013 to 2016, he was a Research Fellow with the School of Electronics and Computer Science, University of Southampton. From 2016 to 2018, he was a Principal Engineer with Huawei Technologies Co. Ltd., where he made significant contributions to the products, solutions, and standardization of 5G, wideband IoT, and cloud gaming/VR. He was a Guest Researcher with the Isaac Newton Institute for Mathematical Sciences, University of Cambridge. He is currently a Full Professor with BUPT. His research interests include 5G/5G-A/6G, massive MIMO, mobile ad hoc networks, distributed artificial intelligence, and cloud gaming/VR. He is a Deputy Director of the Key Laboratory of Mathematics and Information Networks, Ministry of Education, and a Standing Committee Member of the CCF Technical Committee on Distributed Computing and Systems. He received the Dean’s Award for Early Career Research Excellence from the University of Southampton in 2015, the Huawei President Award for Wireless Innovations in 2018, the IEEE TCGCC Best Journal Paper Award in 2019, the IEEE Communications Society Best Survey Paper Award in 2020, the Xiaomi Young Scholars Award in 2023, the CAI Invention and Entrepreneurship Award in 2023, the CIUR Industry-University-Research Cooperation and Innovation Award in 2023, and the First Prize of Beijing Municipal Science and Technology Advancement Award in 2023. He is an Editor of \emph{IEEE Transactions on Communications}, \emph{IEEE Transactions on Vehicular Technology}, and \emph{Signal Processing} (Elsevier), as well as a member of the Editorial Advisory Board of Advanced Computing (Wiley). He was also an Editor of \emph{IEEE Systems Journal} and \emph{IEEE Wireless Communications Letters}. For more details on his research progress, please refer to https://shaoshiyang.weebly.com/.
\end{IEEEbiography}

\begin{IEEEbiography}[{\includegraphics[width=1in,height=1.25in,clip,keepaspectratio]{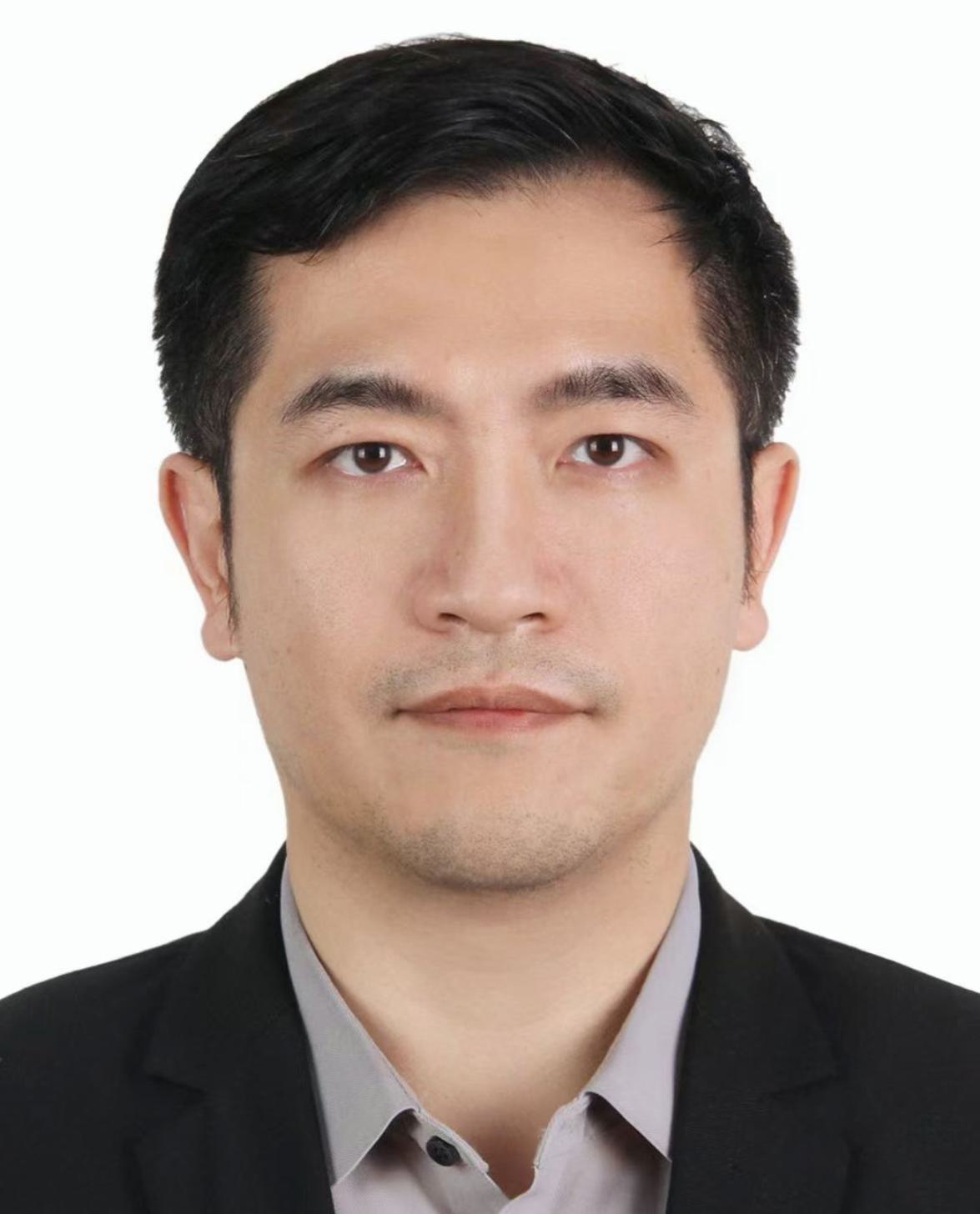}}]{Song Zhao}
	is a senior engineer with the Future Technology Center, China Telecom Research Institute, Beijing, China. He received his Ph.D. degree in Telecommunications and Information Systems from Beijing University of Posts and Telecommunications, Beijing, China. He serves as a delegate of China Telecom in the 3GPP SA2 and SA5 working groups, and as rapporteur for multiple 3GPP work items related to network automation and intelligence. His research interests include wireless channel modeling, radio resource management, proximity services, and network AI.
\end{IEEEbiography}

\begin{IEEEbiography}[{\includegraphics[width=1in,height=1.25in,clip,keepaspectratio]{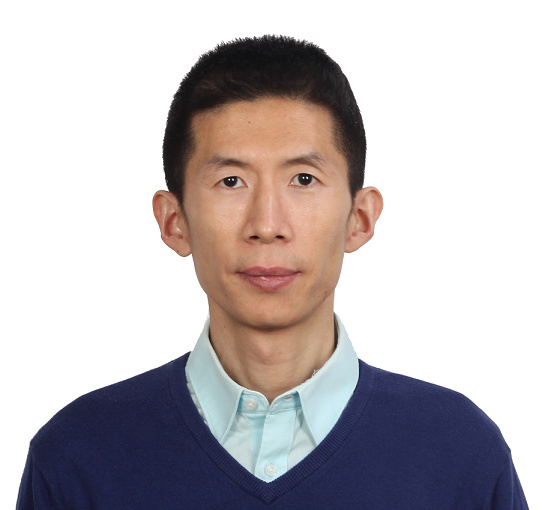}}]{Jinyang Yu}
 received the master’s degree in Communications and Information Systems from Nanjing University of Posts and Telecommunications, Nanjing, China, in 2008. He is currently a Research Director with the China Telecom Research Institute, Beijing, China. His current research interests include space-air-ground integrated networks, ISAC, and mobile AI.
\end{IEEEbiography}

\begin{IEEEbiography}
	[{\includegraphics[width=1in,height=1.25in,clip,keepaspectratio]{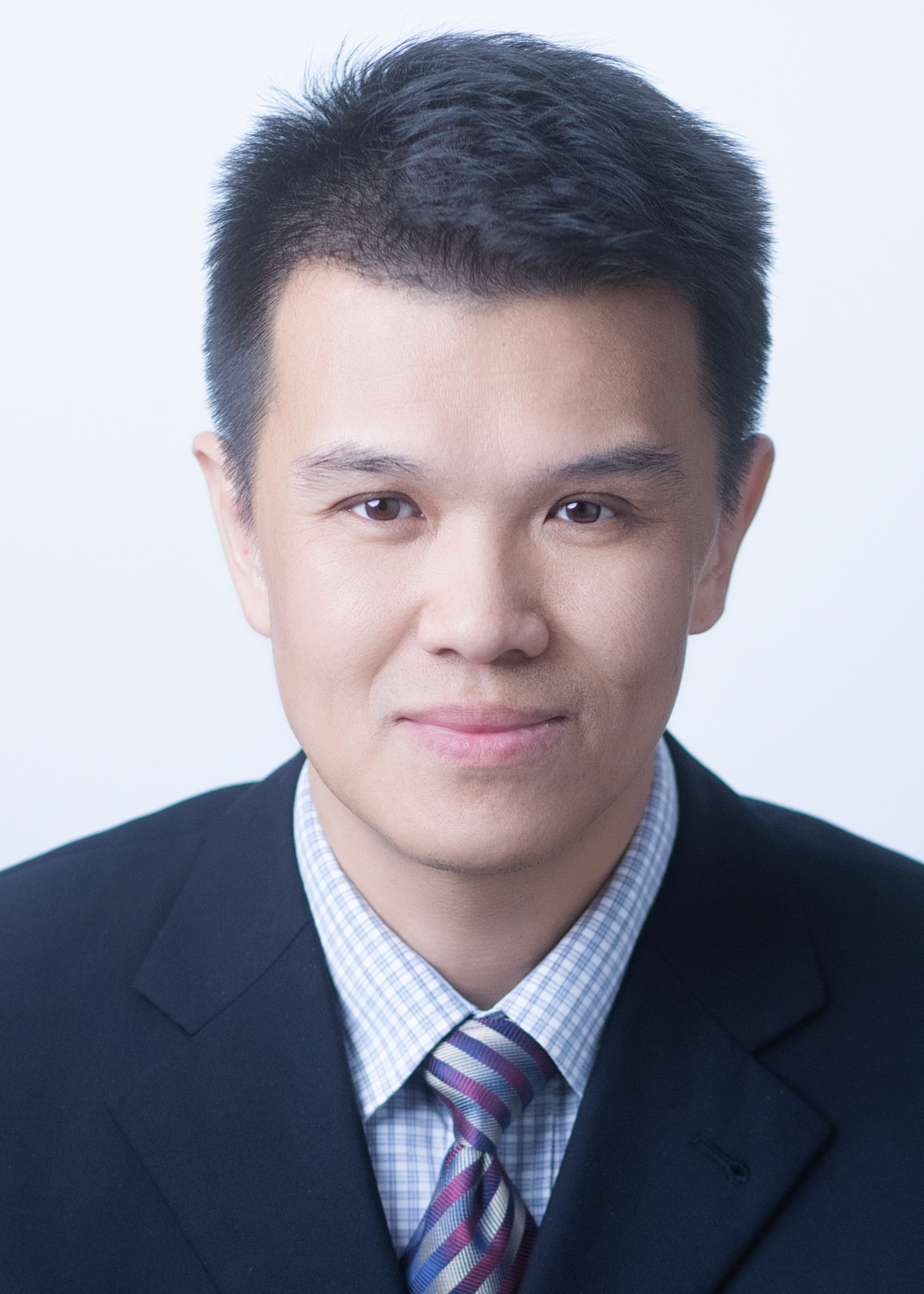}}]
	{Weiliang Xie}
 received the Ph.D. degree in information science and technology from Peking University, China. He is currently a Senior Engineer and Deputy Director of the Mobile Communication Technology Research Department, China Telecom Research Institute. His current research interests include mobile networks and wireless communication systems. He was the principal investigator of the National Science and Technology Major Project of China.
\end{IEEEbiography}

\begin{IEEEbiography}
	[{\includegraphics[width=1in,height=1.25in,clip,keepaspectratio]{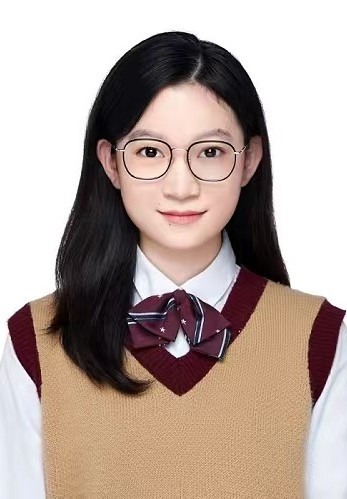}}]
	{Rui-Si Han}
	received her B.S. degree in E-Commerce and Law from Beijing University of Posts and Telecommunications, Beijing, China, in 2021, and her M.Eng. degree in Information and Communication Engineering from the same university in 2024. She is currently a researcher at the China Telecom Cloud Network Operating System R\&D Center, Beijing, China. Her research interests include mobile ad hoc networks and research project management.
\end{IEEEbiography}

\begin{IEEEbiography}[{\includegraphics[width=1in,height=1.25in,clip,keepaspectratio]{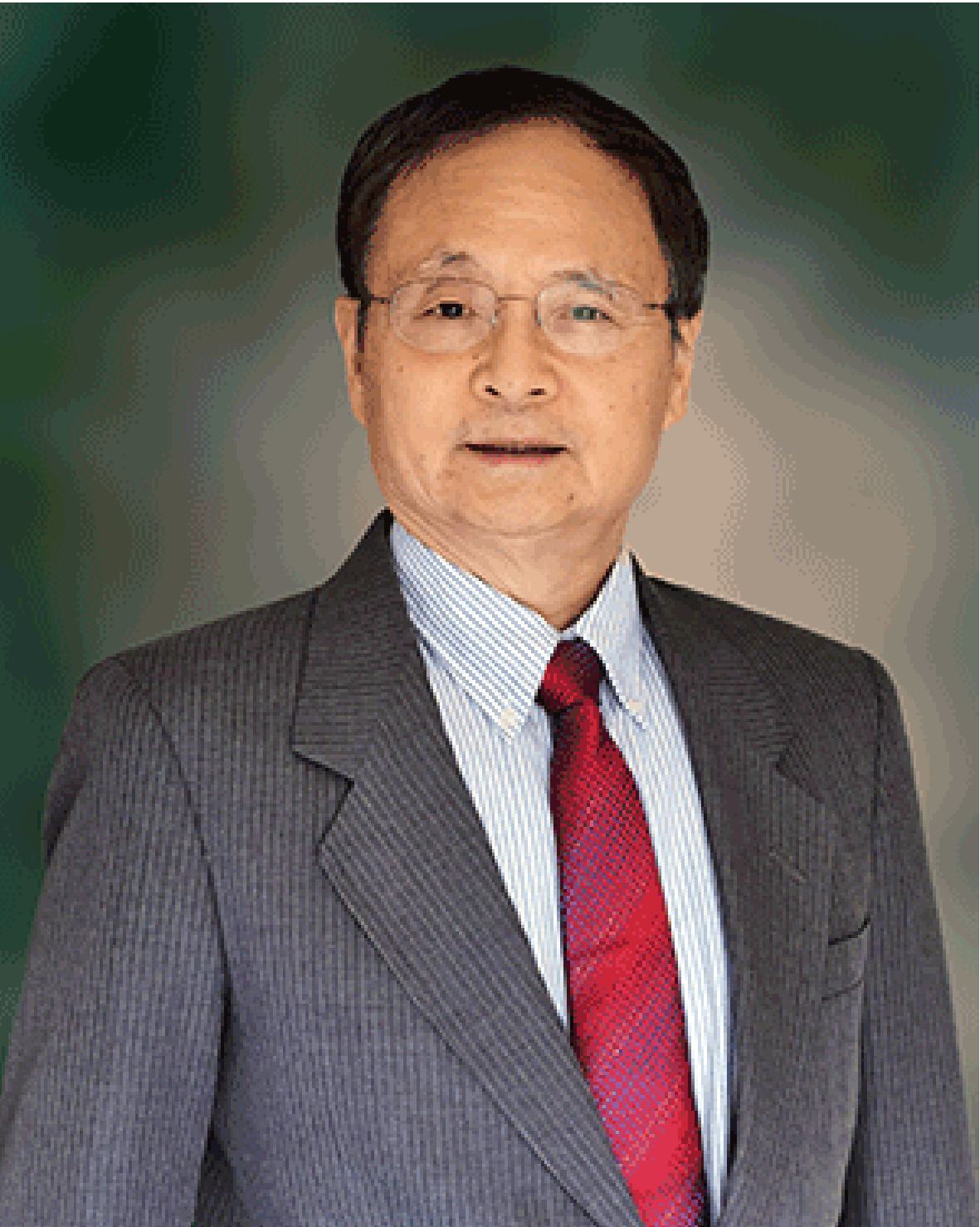}}]{Qi Bi}
	(Fellow, IEEE) is Chief Scientist of China Telecom and CTO of its Research Institute, specializing in 5G and 6G. He received his M.S. from Shanghai Jiao Tong University and Ph.D. from Pennsylvania State University. He was awarded Bell Labs Fellow in 2002, the Bell Labs President’s Gold Awards in 2000 and 2002, the Asian American Engineer of the Year in 2005, the Beijing Outstanding Contribution Award for Innovation and Entrepreneurship for Overseas Scholars in 2019, and three Patent Silver Prizes from the China National Intellectual Property Administration from 2022 to 2024.
\end{IEEEbiography}

\begin{IEEEbiography}[{\includegraphics[width=0.9in,height=1.1in,clip,keepaspectratio]{ShengChen2014.jpg}}]{SHENG CHEN} (IEEE Life Fellow) is an Emeritus Professor, University of Southampton, Southampton, U.K.
	
	He received the B.Eng. degree in control engineering from the East China Petroleum Institute, Dongying, China, in January 1982, the Ph.D. degree in control engineering from City University, London, in September 1986, and the higher doctoral (D.Sc.) degree from the University of Southampton, Southampton, U.K., in August 2005. From 1986 to 1999, he held research and academic appointments with the University of Sheffield, the University of Edinburgh, and the University of Portsmouth, all in U.K. In September 1999, he joined the School of Electronics and Computer Science, University of Southampton, and retired in June 2026. Since then, he is a University Visitor associated with the School of Electronics and Computer Science, University of Southampton. In 2023, Dr. Chen was appointed  a Distinguished Professor at Ocean University of China.
	
	Professor Chen has established an over 40-year distinguished academic career in the fields of adaptive signal processing, wireless communications, neural networks and machine learning. He has authored/coauthored over 750 research papers. He has 23,000+ Web of Science citations with h-index 67, and 44,000+ Google Scholar citations with h-index 88. Professor Chen is IEEE ComSoc Signal Processing and Computing for Communications (SPCC) 2025 Technical Recognition Award recipient.
	
	Professor Chen was elected as a Fellow of the United Kingdom Royal Academy of Engineering in 2014. He is a fellow of Asia-Pacific Artificial Intelligence Association, a fellow of Industry Academy of the International Artificial Intelligence Industry Alliance, a fellow of IET, and one of the original 200 ISI Highly Cited Researchers in engineering (March 2004).
\end{IEEEbiography}

\end{document}